\documentclass[12pt,a4paper]{article}
\usepackage{titling}

\usepackage[T1]{fontenc}
\usepackage[utf8]{inputenc}
\usepackage{lmodern}
\usepackage[margin=2.45cm]{geometry}
\usepackage{amsmath,amssymb,amsthm,mathtools}
\usepackage{enumitem}
\usepackage{microtype}
\usepackage{booktabs}
\usepackage{array}
\usepackage{tikz}
\usetikzlibrary{positioning,calc}
\usepackage{xcolor}
\usepackage{hyperref}

\hypersetup{
  colorlinks=true,
  linkcolor=blue,
  citecolor=blue,
  urlcolor=blue,
  pdftitle={Riesz Energy Subset Selection in the Euclidean Plane Is NP-Hard},
  pdfauthor={Michael Emmerich}
}

\newtheorem{theorem}{Theorem}[section]
\newtheorem{lemma}[theorem]{Lemma}
\newtheorem{proposition}[theorem]{Proposition}

\newtheorem{remark}[theorem]{Remark}

\newcommand{\R}{\mathbb{R}}
\newcommand{\Q}{\mathbb{Q}}
\newcommand{\abs}[1]{\lvert #1\rvert}
\newcommand{\norm}[1]{\lVert #1\rVert}
\newcommand{\RSSP}{\text{\normalfont\scshape Riesz-Subset}}

\newcommand{\sgnspin}{\{-1,+1\}}

\providecommand{\subtitle}[1]{%
  \posttitle{\par\end{center}\vskip 0.5em%
    \begin{center}\large #1\par\end{center}\vskip 0.5em}%
}
\title{Riesz Energy Subset Selection in the Euclidean Plane Is NP-Hard}
\subtitle{A Reduction from the Ising Model on Planar Cubic Graphs}
\author{\parbox{\textwidth}{\centering Michael Emmerich\\[0.75ex]\small Faculty of Information Technology, University of Jyv\"askyl\"a, Finland}}
\date{August 2026}

\begin{document}
\maketitle

\begin{abstract}
We prove that minimum Riesz $s$-energy subset selection in the Euclidean plane is NP-complete already for the fixed exponent $s=2$.  To our knowledge, this is the first Euclidean hardness result for exact Riesz-energy subset selection in which both the ambient dimension and the exponent are fixed. Since the Riesz-Energy minimization is frequently used in low dimensional settings, this result also has practical relevance.   The reduction uses Barahona's planar cubic Ising model with a uniform field.  A spin is encoded by one diagonal of a four-point square.  Axis-aligned selector chains implement ferromagnetic consistency, while a $45^\circ$ terminal geometry yields an antiferromagnetic source interaction.  Rational diagonal perturbations realize the magnetic field, and all remaining interactions are dominated by polynomial separation.  Because $s=2$ and all coordinates are rational, every constructed energy and the decision threshold are rational exactly. 
\end{abstract}


\section{Introduction}\label{sec:intro}

For a finite set $P\subset\R^d$, an integer $k$, and $s>0$, define the Riesz $s$-energy of $S\subseteq P$ by
\[
 E_s(S)=\sum_{\{p,q\}\subset S}\frac1{\norm{p-q}^s}.
\]
The minimum Riesz-energy subset problem asks for a $k$-point subset of minimum energy.  General-metric hardness is known \cite{Pereverdieva2025}; Euclidean hardness in the plane is known when the exponent is part of the input  \cite{EmmerichPereverdievaDeutz2026}, but standard reductions to unit disk problems do not close for a fixed exponent $s$.  The latter work therefore identifies fixed-exponent Euclidean hardness as a principal open case.  The purpose of the present paper is to prove NP-completeness for fixed geometric parameters at
\[
 d=2,\qquad s=2.
\]

\begin{theorem}[Main theorem]\label{thm:main}
The decision problem $\RSSP(2,2)$ is NP-complete.  More precisely, from an instance $(G,K)$ of maximum independent set on a planar cubic graph one can construct in polynomial time a finite rational point set $P\subset\Q^2$, an integer $k$, and a rational threshold $\tau$ such that
\[
 \exists S\subseteq P,\quad |S|=k,\quad E_2(S)\le\tau
\]
if and only if $G$ has an independent set of cardinality at least $K$.
\end{theorem}

To our knowledge, Theorem~\ref{thm:main} is the first NP-hardness result for exact Riesz-energy subset selection in a Euclidean space in which both the ambient dimension and the exponent are fixed.  The fixed-dimensional aspect is important: on ordered points of the real line, polynomial-time algorithms are available for every fixed $s>0$ through the Monge/submodular structure of the one-dimensional interaction \cite{EmmerichLine2026}.  The theorem therefore identifies a genuine complexity change between the ordered line and the Euclidean plane at the same fixed exponent.

The source problem, i.e., the problem that is reduced in polynomial time to Riesz-energy subset selection, is the planar Ising model with magnetic field appearing in Barahona's classical hardness proof \cite{Barahona1982}.  Barahona starts from the decision version of \emph{Independent Set} restricted to planar cubic graphs, which is NP-complete.  The source Ising model and the target Riesz problem involve different notions of choice.  In the source problem, each vertex $v$ is assigned a spin variable $\sigma_v\in\{-1,+1\}$, and one minimizes an energy function, i.e. the source Hamiltonian $H_B$ defined in Section~\ref{sec:source}, that depends on the fixed interaction graph and on the chosen spin assignment.   Figure~\ref{fig:ising-cardinality-example} illustrates the Ising problem  on the planar cubic graph~$Q_3$. The Ising problem does not involve fixed-cardinality subset selection; it only requires assigning values to the  spin variables.  Fixed-cardinality subset selection enters only in the Riesz-energy target problem, where one minimizes a different energy function over subsets of prescribed size.

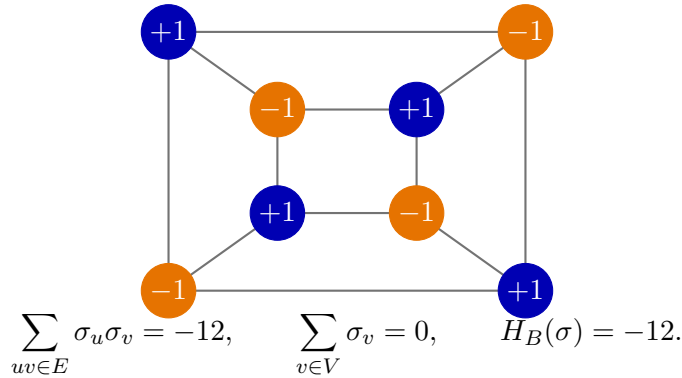
\begin{figure}[htbp]
\centering
\begin{tikzpicture}[
    x=1.18cm,y=1.18cm,
    edge/.style={draw=black!55,line width=0.8pt},
    plus/.style={circle,draw=blue!70!black,fill=blue!70!black,
                 text=white,font=\small\bfseries,minimum size=7.2mm,
                 inner sep=0pt},
    minus/.style={circle,draw=orange!90!black,fill=orange!90!black,
                  text=white,font=\small\bfseries,minimum size=7.2mm,
                  inner sep=0pt}
]
  \coordinate (v0) at (-2.0, 1.45);
  \coordinate (v1) at ( 2.0, 1.45);
  \coordinate (v2) at ( 2.0,-1.45);
  \coordinate (v3) at (-2.0,-1.45);
  \coordinate (v4) at (-0.78, 0.58);
  \coordinate (v5) at ( 0.78, 0.58);
  \coordinate (v6) at ( 0.78,-0.58);
  \coordinate (v7) at (-0.78,-0.58);

  \draw[edge] (v0)--(v1)--(v2)--(v3)--cycle;
  \draw[edge] (v4)--(v5)--(v6)--(v7)--cycle;
  \draw[edge] (v0)--(v4);
  \draw[edge] (v1)--(v5);
  \draw[edge] (v2)--(v6);
  \draw[edge] (v3)--(v7);

  \node[plus]  at (v0) {$+1$};
  \node[minus] at (v1) {$-1$};
  \node[plus]  at (v2) {$+1$};
  \node[minus] at (v3) {$-1$};
  \node[minus] at (v4) {$-1$};
  \node[plus]  at (v5) {$+1$};
  \node[minus] at (v6) {$-1$};
  \node[plus]  at (v7) {$+1$};

  \node[align=center,font=\small] at (0,-2.05)
    {$\displaystyle \sum_{uv\in E}\sigma_u\sigma_v=-12,\qquad
       \sum_{v\in V}\sigma_v=0,\qquad H_B(\sigma)=-12.$};
\end{tikzpicture}
\caption{A ground-state assignment of the source Hamiltonian \eqref{eq:HB} on the planar cubic cube graph~$Q_3$.  Blue vertices carry spin $+1$ and orange vertices spin $-1$.  Every edge joins opposite spins, so all twelve antiferromagnetic edge terms contribute $-1$, while the uniform field term sums to zero.  Hence $H_B=-12=|V|/2-4\alpha(Q_3)$.  The four $+1$ spins correspond to a maximum independent set; importantly, their number is an outcome of energy minimization, not a fixed-cardinality constraint on the Ising problem.}
\label{fig:ising-cardinality-example}
\end{figure}

In the geometric reduction, each vertex of the Ising instance is represented by a three-branched selector tree routed along orthogonal directions in the plane.  Ferromagnetic interactions along these tree branches enforce a common spin value, so that the entire tree represents a single Ising variable.  Each edge of the source Ising graph is then realized by an objective interaction between terminal selectors arranged at a $45^\circ$ angle, producing the required antiferromagnetic coupling between the corresponding vertex spins.
Each four-point selector realizes either spin state by selecting one of its two diagonals, and both states contain exactly two Riesz points. Consequently, a spin flip changes the logical state without changing the number of selected points, while the global target cardinality is fixed at twice the number of selectors.

A central idea in the proof is that we use a \emph{diagonal selector} consisting of four points of the Riesz-energy input. Points are forced to be selected on diagonals rather than on horizontals by Riesz-energy minimization, and 'flipping' the diagonal translates to 'flipping' the spin in the Ising problem. The four candidate points form a very small square, and the two logical states are its two diagonals.  The selector has an angular multipole law.  Two axis-aligned neighboring selectors induce a field-free ferromagnetic coupling favoring same spin interaction, whereas a center vector at angle $45^\circ$ induces an antiferromagnetic coupling which favors opposite spin interactions.  Thus the same local object implements ferromagnetic consistency wires to implement source problem vertices (when used at 90$^\circ$ angles) and antiferromagnetic source edge interactions (when used at 45$^\circ$ angles) to implement source problem edges.  At exponent two the spin--spin term decays as $r^{-6}$, which remains summable with ample slack at bends and degree-three junctions.

A useful simplification is specific to $s=2$.  For rational points $p,q\in\Q^2$,
\[
 \norm{p-q}^{-2}=\frac1{(p_1-q_1)^2+(p_2-q_2)^2}\in\Q.
\]
Hence the field compensation and the final threshold can be performed with exact rational arithmetic.  This exactness also yields membership in NP by allowing a proposed subset energy to be evaluated and compared with the threshold using rational arithmetic of polynomial bit length.

Although the present result is complexity-theoretic rather than an application-specific hardness theorem, it may help clarify computational limits in several domains where Riesz-type repulsion is used to promote separation, diversity, or coverage.  In the plane, Riesz-energy minimization has been studied directly for facility location, and closely related repulsive criteria arise in spatial coverage and representative-site placement \cite{Atta2026}.  In multi- and many-objective optimization, Riesz $s$-energy has been used to construct well-spaced weight vectors and reference directions on the unit simplex \cite{BlankDebDhebarBandaruSeada2021}, as well as diverse reference sets for Pareto optimization when a simplex-lattice geometry is not appropriate \cite{FalconCardonaIshibuchiCoello2020}.  In potential theory, minimum Riesz-energy configurations model repelling particle distributions and provide a principled way to discretize rectifiable manifolds by asymptotically well-distributed point sets \cite{HardinSaff2005,HardinSaff2004}.  These formulations are not identical to the fixed-cardinality subset problem studied here, so Theorem~\ref{thm:main} should not be read as a hardness statement for each of them.  Rather, it shows that exact subset selection under the same basic repulsive potential can already be computationally difficult in the natural low-dimensional setting of the Euclidean plane, which supports the practical importance of approximation methods and fast heuristics.


\paragraph{Proof organization.}
Section~\ref{sec:walkthrough} gives a conceptual walkthrough of the reduction before the quantitative estimates are introduced.  Section~\ref{sec:source} fixes the planar cubic Ising source and its constant gap.  Section~\ref{sec:selector} shows how the global cardinality constraint forces exactly one binary diagonal state in every selector.  Sections~\ref{sec:multipole} and \ref{sec:field} derive the angular two-spin interaction and compensate the resulting one-spin fields by exact rational perturbations.  Sections~\ref{sec:routing} and \ref{sec:objective} place the selector trees in an orthogonal planar routing and realize each source edge by a $45^\circ$ objective box.  Section~\ref{sec:consistency} then proves, by energy-decreasing component flips, that every consistency tree can be reduced to one common spin.  Finally, Sections~\ref{sec:normalized} through \ref{sec:encoding} collapse the normalized geometry to the scaled Barahona Hamiltonian, choose the rational threshold, and verify polynomial encoding and NP membership.  The appendices provide a worked finite example, a detailed multipole derivation, and the auxiliary geometric bounds used in the proof.

\section{How the reduction works}\label{sec:walkthrough}

This section explains the construction before the quantitative estimates begin.  Starting from a planar cubic graph, Section~\ref{sec:source} associates one spin with each source vertex and uses Barahona's Hamiltonian to encode maximum independent set.  The geometric reduction replaces each source spin by a tree of four-point selector cells.  Axis-aligned neighboring cells force that tree to carry one common spin, and a $45^\circ$ terminal arrangement between two trees produces the interaction corresponding to a source edge.  Small rational changes in the $+$ diagonals produce the required one-spin fields.  The later estimates prove that selector defects can always be repaired, that consistency defects can always be repaired, and that all remaining long-range interactions are too small to cross the fixed source-energy gap.

\subsection*{From a cubic graph to an Ising system}

Let $G$ be a planar cubic graph.  We attach one binary spin to every vertex.  We interpret one spin state as ``selected'' and the other as ``not selected.''  The Ising Hamiltonian used by Barahona contains an antiferromagnetic interaction on every graph edge together with the same one-spin field at every vertex.

It is useful to be precise about the role of the field.  With the sign convention used in this paper, the field term by itself does not literally reward the selected state.  Rather, because every source vertex has degree three, the three incident edge terms and the one local field term combine in a particularly simple way.  If $t$ vertices are declared selected and $q$ graph edges have both endpoints selected, then, up to a fixed additive constant, the Ising objective is
\[
 4(q-t).
\]
Thus adding a selected vertex gives a benefit of four units, while every edge whose two endpoints are both selected costs four units.  Any selected set containing an internal edge can therefore be repaired without losing objective value, and the ground state encodes a maximum independent set.  Section~\ref{sec:source} gives the exact identity and the resulting gap.

\subsection*{One selector box is one binary variable}

A logical spin is represented by four candidate points at the corners of a very small square.  The global cardinality is chosen so that, after the forcing argument, exactly two points are selected from every square.  Adjacent corners are closer than opposite corners, so a side pair has noticeably larger Riesz energy than a diagonal pair.  Consequently the only low-energy two-point choices are the two diagonals.

We designate one diagonal as spin $+1$ and the other as spin $-1$.  In this sense a selector box is a geometric binary variable.  The selector-forcing penalty is intentionally enormous compared with every later energy scale.  It is of order $b^{-2}$, while the logical spin interactions are of order $b^4$.  This first scale separation ensures that the optimization problem cannot profit by using the wrong number of points in a box or by choosing a side instead of a diagonal.

\subsection*{How a cubic vertex becomes a selector tree}

A cubic source vertex has exactly three incident edges.  After taking a crossing-free orthogonal drawing of the planar graph, we replace one source vertex by one rooted selector tree $T_v$.  Topologically, this tree is extremely simple.  It is a three-leaf star.  There is one root selector and three internally disjoint rectilinear branches, one for each incident graph edge.  The branches may be long and may contain right-angle bends, but there is no recursive branching and no exponential gadget growth.

Along each branch, selector centers are placed at unit spacing on horizontal or vertical segments.  For such axis-aligned neighboring selectors the effective Riesz coupling is ferromagnetic, so equal selector states have lower energy than unequal states.  The chain therefore acts as a wire carrying one binary value through the planar drawing.  Bends and the degree-three junction at the root are handled by the same consistency analysis.

Only one designated selector, the root of $T_v$, is assigned the nonzero target one-spin field.  The dummy selectors along the wires have target field zero.  Small rational perturbations of the selector diagonals cancel the unwanted geometric one-spin fields and leave precisely this intended pattern.  Thus the field is a local part of the Ising encoding; it is not implemented by a separate long-range gadget.

\subsection*{How a graph edge becomes a $45^\circ$ objective box}

For every source edge $uv$, one branch of $T_u$ and one branch of $T_v$ terminate in a small objective box.  Their terminal prefixes are arranged so that the unique closest cross-tree pair of selector centers is separated by a $45^\circ$ vector.  At this angle the sign of the effective selector coupling is opposite to the axis-aligned wire coupling and is antiferromagnetic.

The distinguished $45^\circ$ pair is accompanied by a halo of more distant cross-prefix interactions.  We do not pretend that those interactions vanish.  Instead we bound their complete absolute contribution and prove that the positive distinguished interaction dominates the halo.  Every source edge therefore realizes the same positive objective coefficient $\lambda$.  After the consistency trees have been normalized, that box contributes exactly the intended Ising interaction between the logical spins of $u$ and $v$, up to the globally controlled perturbation terms.

\newpage
\subsection*{The hierarchy of penalties}

The construction is organized around a strict hierarchy of energy scales.  It is helpful to keep the following four levels in mind.

\begin{enumerate}[label=\textbf{Level \arabic*.},leftmargin=3.2em]
\item \textbf{Selector legality.}  The very large intra-box penalty forces exactly one diagonal to be chosen in each four-point selector.  This scale is of order $b^{-2}$ and dominates everything else.

\item \textbf{Tree consistency.}  Once every box is a valid spin, the strongest logical interaction is the nearest-neighbor axis coupling along the selector trees.  A disagreeing wire edge can be repaired by flipping the entire component on the leaf side of that edge.  The gain from repairing the cut edge is larger than the sum of all adverse same-tree tails, one possible objective-box load, and all residual and remote terms.  Hence every energy minimum has a monochromatic tree $T_v$.

\item \textbf{The intended source Hamiltonian.}  After normalization, each whole tree represents a single spin.  The objective box on every source edge and the root field at every source vertex are deliberately matched to the same coefficient $\lambda$.  These are not separate asymptotic scales.  Together they reproduce Barahona's Ising Hamiltonian.

\item \textbf{Remainder terms.}  All unintended long-range interactions, residual fields, and coupling changes caused by the small rational perturbations are forced below a small fraction of $\lambda$.  They are therefore too weak to change the yes/no decision inherited from the fixed source gap.
\end{enumerate}

This hierarchy separates two logically different tasks.  Levels 1 and 2 prove \emph{normalization}, which turns an arbitrary feasible Riesz subset into a configuration with one valid spin per selector and one common spin per source tree.  Levels 3 and 4 prove \emph{simulation}.  For a normalized configuration, the exact energy equals a constant baseline plus $\lambda H_B(\sigma)$ and a remainder whose absolute value is eventually bounded by $\lambda/20$.

\subsection*{Why long wires do not destroy the reduction}

Riesz energy is long range because every selected point interacts with every other selected point.  Hence a long selector tree cannot simply be treated as a nearest-neighbor spin chain.  The proof controls this in two different ways.

First consider interactions inside one tree.  If a defective consistency edge is cut and the component away from the root is flipped, point pairs across that cut occur at many separations.  At spin level the unwanted coupling decays like the sixth power of the distance.  Along a one-dimensional wire, the number of cross-cut pairs at separation $k$ grows only linearly in $k$, so the complete tail is bounded by a convergent sum of the form $\sum k^{-5}$.  Right-angle bends and the degree-three junction lead to similarly convergent lattice sums.  Crucially, these bounds are uniform and do not grow with the total length of the branch.

Second, geometrically unrelated features are placed far apart.  The orthogonal drawing is blown up by a polynomial separation scale $D$.  Even if we pessimistically count all $M^2$ possible remote selector pairs, each remote spin-spin coupling is at most a constant times $b^4D^{-6}$.  The chosen polynomial scale makes the complete remote load tiny compared with one intended objective coefficient.  Thus increasing the wire length only increases $M$ polynomially; the separation is chosen large enough that the total far field still decreases to a negligible level.

The three-leaf topology gives one further simplification in the component-flip proof.  Deleting a consistency edge and taking the component away from the root can expose at most one source-edge objective box.  Therefore a repair never has to pay for several objective interactions at once.

\subsection*{Collapsing a tree to one bit}

The word ``collapse'' refers to a proof step, not to a geometric operation.  Before normalization, a tree containing many selectors has many possible spin assignments.  The component-flip lemma proves that every assignment containing a consistency defect can be strictly improved.  Hence only two relevant states remain.  All selectors in $T_v$ are $+1$, or all are $-1$.

At that point we replace the entire tree in the analysis by one logical variable $\sigma_v$.  Every same-tree spin product is then constant and can be absorbed into a baseline energy.  What remains is one root field for each source vertex, one objective interaction for each source edge, and a small remainder.  In other words, the large geometric instance reduces to the original Ising system at the level of its normalized states.

The final correctness argument can therefore be read as one continuous implication.  A planar cubic independent-set instance determines a planar Ising ground-state problem.  The geometric construction realizes each Ising spin by a normalized selector tree and each Ising edge by one objective box.  Once the normalization lemmas have been applied, comparing Riesz $2$-energies is equivalent to comparing the corresponding scaled Ising energies up to the controlled remainder.  That remainder is made much smaller than the fixed gap between consecutive independent-set values, so it cannot change which side of the decision threshold the instance lies on.

\section{The planar cubic source and its gap}\label{sec:source}

Let $G=(V,E)$ be planar and cubic, with $n=|V|$ and $m=|E|=3n/2$.  The threshold decision problem
\[
  \text{``does $G$ contain an independent set of size at least $K$?''}
\]
is NP-complete on this graph class.  Garey, Johnson, and Stockmeyer proved NP-completeness of the equivalent vertex-cover problem for planar graphs of maximum degree three \cite{GareyJohnsonStockmeyer1976}; minimum vertex cover remains NP-hard even on 2-connected cubic planar graphs \cite{Mohar2001}, using $\alpha(G)+\operatorname{vc}(G)=|V|$; and Uehara showed that the independent-set problem stays NP-complete under the additional assumptions of 3-connectivity and girth greater than three \cite{Uehara1996}.

Barahona formulates the source optimization problem as P4, maximum stable set in a planar cubic graph, and then defines P5 on a planar graph with all interactions antiferromagnetic and uniform magnetic field $1$ \cite{Barahona1982}.  In Barahona's reduction from P4 to P5 the Ising instance uses the same graph $G$.  Consequently P5 remains NP-hard even when its interaction graph is restricted to be planar and cubic.  Rather than appeal to that statement as a black box, we now write the many-one reduction to the threshold form explicitly.  We use the equivalent Hamiltonian
\begin{equation}\label{eq:HB}
 H_B(\sigma)=\sum_{uv\in E}\sigma_u\sigma_v+\sum_{v\in V}\sigma_v,
 \qquad \sigma_v\in\sgnspin.
\end{equation}

\begin{proposition}[Explicit source reduction]\label{prop:source-reduction}
Given a planar cubic Independent Set instance $(G,K)$, let
\[
 H_0=\frac{|V|}{2}-4K.
\]
Then $G$ has an independent set of size at least $K$ if and only if there exists a spin assignment $\sigma\in\{-1,+1\}^{V}$ with $H_B(\sigma)\le H_0$.  Membership in NP is immediate because $H_B(\sigma)$ is an integer sum that can be evaluated in polynomial time.  Hence this restricted Ising threshold problem is NP-complete, and its optimization version is NP-hard.
\end{proposition}

For completeness, set
\[
 x_v=\frac{1+\sigma_v}{2}\in\{0,1\},\qquad
 t=\sum_v x_v,
\]
and let $q$ be the number of edges with both endpoints in $\{v:x_v=1\}$.  Since $G$ is cubic,
\[
 \sum_{uv\in E}(x_u+x_v)=3t.
\]
Substitution into \eqref{eq:HB} gives
\begin{equation}\label{eq:barahona-id}
 H_B(\sigma)=\frac n2+4(q-t).
\end{equation}
Every $t$-vertex set inducing $q$ edges contains an independent set of size at least $t-q$.  Choose one endpoint from every induced edge and delete the union of those at most $q$ chosen vertices.  Conversely an independent set of size $\alpha(G)$ has $q=0$ and $t=\alpha(G)$.  Hence
\begin{equation}\label{eq:barahona-min}
 \min_\sigma H_B(\sigma)=\frac n2-4\alpha(G).
\end{equation}
In particular, for an independent-set target $K$,
\begin{align}
 \alpha(G)\ge K&\Longrightarrow \min H_B\le H_0,\label{eq:source-yes}\\
 \alpha(G)\le K-1&\Longrightarrow \min H_B\ge H_0+4,\label{eq:source-no}
\end{align}
where
\begin{equation}\label{eq:H0}
 H_0=\frac n2-4K.
\end{equation}
The fixed gap $4$ will absorb all geometric approximation errors.

\section{Diagonal selector cells}\label{sec:selector}

\subsection{The unperturbed selector}

For a center $c\in\R^2$ and scale $b>0$, define the four candidate points
\[
 c+(b,b),\quad c+(b,-b),\quad c+(-b,b),\quad c+(-b,-b).
\]
The two intended states are the diagonals
\begin{align*}
 S_+(c)&=\{c+(b,b),c-(b,b)\},\\
 S_-(c)&=\{c+(b,-b),c-(b,-b)\}.
\end{align*}
We identify them with spins $\sigma=+1$ and $\sigma=-1$.

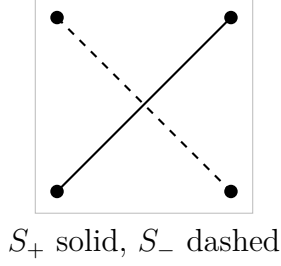
\begin{figure}[ht]
\centering
\begin{tikzpicture}[scale=1.15]
  \draw[gray!50] (-1.25,-1.25) rectangle (1.25,1.25);
  \fill (-1,-1) circle (2.2pt); \fill (1,1) circle (2.2pt);
  \fill (-1,1) circle (2.2pt); \fill (1,-1) circle (2.2pt);
  \draw[thick] (-1,-1)--(1,1);
  \draw[thick,dashed] (-1,1)--(1,-1);
  \node at (0,-1.6) {$S_+$ solid, $S_-$ dashed};
\end{tikzpicture}
\caption{A planar diagonal selector. Exactly two of the four candidate points will be chosen after the forcing argument below.}
\label{fig:selector}
\end{figure}

For $j\in\{0,1,2,3,4\}$, let $I_j$ denote the minimum internal Riesz $2$-energy among choices of $j$ corners of one unperturbed selector cell.  Only pairs of selected points contribute to this internal energy.  The quantity $I_j$ therefore depends only on the pair distances inside one square and does not include interactions with points in other selector cells.

The cases $j=0$ and $j=1$ are important even though their value is zero.  If no corner is selected, there is no pair.  If exactly one corner is selected, there is still no pair.  In both cases the defining pair sum is empty, and therefore
\[
 I_0=I_1=0.
\]
For $j=2$ there are six possible two-corner choices.  Four are sides of length $2b$, and two are diagonals of length $2\sqrt2\,b$.  Since the inverse-square energy decreases when the distance increases, a diagonal has the smaller energy and hence realizes the minimum $I_2$.  Thus
\begin{equation}\label{eq:Idiag}
 I_2=(2\sqrt2 b)^{-2}=\frac1{8}b^{-2}.
\end{equation}
A side pair has distance $2b$ and energy $\frac14b^{-2}$.  The energy loss incurred by choosing a side instead of a diagonal is therefore
\begin{equation}\label{eq:csel}
 \Delta_{\rm sel}=c_{\rm sel}b^{-2},\qquad
 c_{\rm sel}=\frac18.
\end{equation}

For $j=3$, any three corners of a square form the same distance pattern.  Among the three selected pairs, two are sides and one is a diagonal.  Hence every three-corner choice has the same energy, which is $I_3$.  For $j=4$ all four corners are selected, so all six internal pairs contribute.  Four of those pairs are sides and two are diagonals.  Consequently
\begin{align}
 I_3&=2\left(\frac14b^{-2}\right)+\frac18b^{-2}
     =\frac58b^{-2},\label{eq:I3}\\
 I_4&=4\left(\frac14b^{-2}\right)+2\left(\frac18b^{-2}\right)
     =\frac54b^{-2}.
\end{align}
The five values $I_0,\ldots,I_4$ are not separate gadget states.  They are bookkeeping quantities used to prove that a globally feasible subset cannot lower its energy by concentrating too many selected points in some cells and too few in others.  If one selected point is transferred from an overfull cell to an underfull cell, the possible occupancy patterns are $(3,1)\to(2,2)$, $(3,0)\to(2,1)$, $(4,1)\to(3,2)$, and $(4,0)\to(3,1)$.  Their unperturbed internal improvements are respectively $3b^{-2}/8$, $b^{-2}/2$, $b^{-2}/2$, and $5b^{-2}/8$.  The smallest improvement is therefore attained by
\begin{equation}\label{eq:balance-gap}
 (3,1)\longrightarrow(2,2),\qquad
 I_3-2I_2=\frac38b^{-2}.
\end{equation}
Here the entries $3$ and $1$ are the numbers of selected points in the two cells before the transfer.  This is the worst balancing move only in the sense that it yields the smallest guaranteed local energy decrease.

\subsection{Global two-of-four forcing}

Suppose there are $M$ selector cells whose centers are mutually at distance at least $1$, and set the global cardinality to
\begin{equation}\label{eq:k2M}
 k=2M.
\end{equation}
The quantity $M$ is the total number of selector cells in the constructed Riesz instance.  The condition $k=2M$ fixes only the average occupancy at two points per cell.  It does not by itself rule out patterns such as $(3,1)$ or $(4,0)$ across two cells.

The purpose of the next lemma is to turn this global cardinality condition into the intended local selector structure.  It shows that every size-$2M$ subset with an invalid cell can be changed, without altering its cardinality, so that the total Riesz energy strictly decreases.  Repeating the improvement first forces occupancy two in every cell and then replaces every remaining side pair by a diagonal.  The conclusion is local.  Different selector cells may still choose different diagonals.  What has been established is that every cell represents one well-defined spin in $\{-1,+1\}$.  The later consistency-tree argument is responsible for forcing selected groups of cells to agree.

The proof compares a local gain with a global disturbance.  A repair inside one or two cells gains an amount proportional to $b^{-2}$ from the short intra-cell distances.  The same repair also changes interactions with selected points in all other cells.  Because different cell centers are at least unit distance apart, that adverse cross-cell change is only of order $M$.  Condition~\eqref{eq:force-cond} deliberately asks the local $b^{-2}$ gain to dominate this much coarser $M$ bound.  This separation is the reason the lemma remains valid even after the small field-compensation perturbations are introduced.

Section~\ref{sec:field} later perturbs the $+$ diagonal by a very small amount in order to compensate one-spin fields.  The perturbation is not needed for selector forcing itself.  It is included in the statement below so that the forcing property is already valid for the final perturbed instance.

\begin{lemma}[Two-of-four forcing]\label{lem:selector-force}
Let $M$ selector cells have centers that are pairwise at distance at least $1$, let $k=2M$, and assume $b\le10^{-3}$ and that every cell $i$ has the perturbed form of Section~\ref{sec:field} with $\abs{\eta_i}\le10^{-6}b$.  If
\begin{equation}\label{eq:force-cond}
 \frac1{16}b^{-2}>4M,
\end{equation}
then every size-$2M$ subset that is not a valid diagonal in every cell admits a cardinality-preserving move that strictly decreases Riesz $2$-energy.
\end{lemma}

\begin{proof}
Every perturbed candidate remains within $2b$ of its cell center.  Distinct cells therefore have point-to-point distance at least $1-4b>0.996$.  Adding one point creates less than $2M(1-4b)^{-2}<3M$ adverse cross-cell energy.  We use the round upper bound $4M$ throughout the proof.

Suppose first that some cell has occupancy below two.  Since the total cardinality is $2M$, another cell has occupancy above two.  Transfer one selected point from an overfull cell to an underfull cell, choosing the removal and insertion so that the resulting two internal energies are minimal.  Before perturbation the smallest possible internal improvement is $3b^{-2}/8$ by \eqref{eq:balance-gap}.  If $|\eta|\le10^{-6}b$, every affected internal distance changes by $O(10^{-6}b)$.  Since $|(r^{-2})'|=2r^{-3}$, the change in the internal transfer gain is less than $10^{-3}b^{-2}$ after summing the at most six affected internal pairs.  More than $b^{-2}/16$ of the local gain therefore remains.  By \eqref{eq:force-cond}, this gain is larger than the possible adverse cross-cell contribution, so the total Riesz energy strictly decreases.

Repeat such transfers until every cell has occupancy two.  A two-point cell is still invalid if its selected points form a side.  Replacing a side pair by a diagonal has unperturbed gain $b^{-2}/8$.  The same derivative estimate leaves at least $b^{-2}/16$ after perturbation.  Equation~\eqref{eq:force-cond} again makes the total move strictly energy decreasing.  Repetition terminates because only finitely many size-$2M$ subsets exist.
\end{proof}

\begin{remark}
With the final choice $b=(10^8M^3)^{-1}$, the left-hand side of \eqref{eq:force-cond} is exactly
\[
 \frac{10^{16}}{16}M^6,
\]
while the right-hand side is only $4M$.  Their ratio is
\[
 \frac{10^{16}}{64}M^5,
\]
which is already larger than $10^{14}$ for $M\ge1$ and grows as $M^5$.  The forcing condition therefore has far more slack than the proof needs.  This is intentional.  Lemma~\ref{lem:selector-force} is a normalization statement.  Its task is to prove the direction of every repair move, not to optimize the selector size or the numerical constants.  The conservative bounds $4M$ and $b^{-2}/16$ may both be weakened substantially without affecting the argument.  The perturbations introduced in Section~\ref{sec:field} are smaller still, so they cannot come close to reversing a selector-forcing move.
\end{remark}

\section{Exact two-spin decomposition and the angular multipole law}\label{sec:multipole}

Let
\[
 u_+=(1,1),\qquad u_-=(1,-1).
\]
For two unperturbed selectors with center difference $R\ne0$, define
\begin{equation}\label{eq:Fst}
 F_R(\sigma,\tau)
 =\sum_{\epsilon,\eta\in\sgnspin}
 \norm{R+b(\eta u_\tau-\epsilon u_\sigma)}^{-2}.
\end{equation}
Every function on two binary spins has the exact Fourier representation
\begin{equation}\label{eq:fourier}
 F_R(\sigma,\tau)=A_R+h_R^{(1)}\sigma+h_R^{(2)}\tau+J_R\sigma\tau,
\end{equation}
where
\begin{align}
 A_R&=\tfrac14(F_{++}+F_{+-}+F_{-+}+F_{--}),\nonumber\\
 h_R^{(1)}&=\tfrac14(F_{++}+F_{+-}-F_{-+}-F_{--}),\nonumber\\
 h_R^{(2)}&=\tfrac14(F_{++}-F_{+-}+F_{-+}-F_{--}),\label{eq:fourier-coeff}\\
 J_R&=\tfrac14(F_{++}-F_{+-}-F_{-+}+F_{--}).\nonumber
\end{align}
Write $r=\norm R$, $R/r=(\cos\theta,\sin\theta)$, and $t=b/r$.

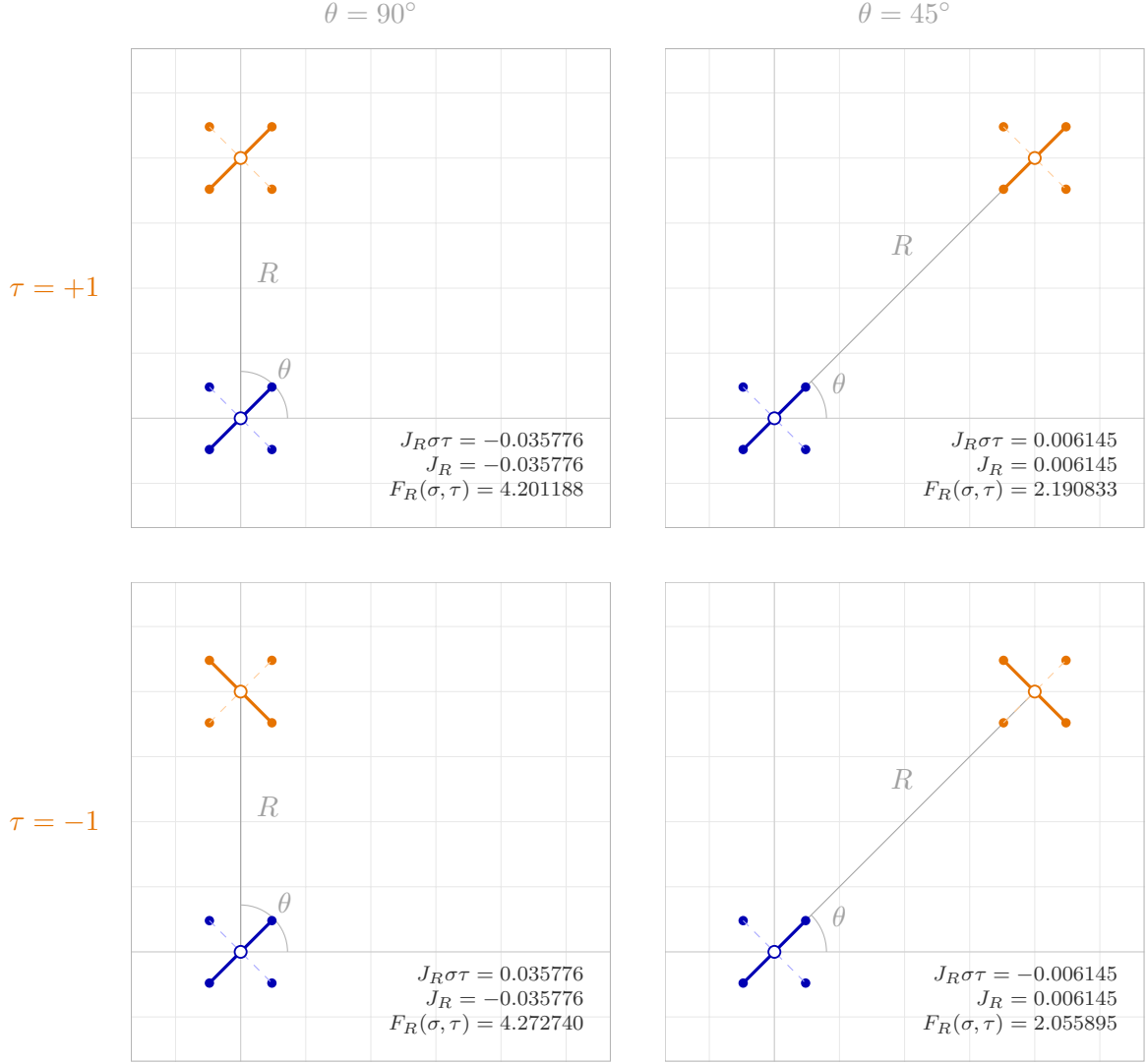
\begin{figure}[htbp]
\centering
\begin{tikzpicture}[x=3.5cm,y=3.5cm,font=\small]
\def\b{0.12}

\newcommand{\selectorpanel}[3]{%
\begin{scope}[shift={(#1,#2)}]
  \pgfmathsetmacro{\aa}{#3}

  \ifnum#3=0
    \def\Jtext{-0.035776}
    \ifnum\taustate=1
      \def\FRtext{4.201188}
      \def\JJtext{-0.035776}
    \else
      \def\FRtext{4.272740}
      \def\JJtext{0.035776}
    \fi
  \else
    \def\Jtext{0.006145}
    \ifnum\taustate=1
      \def\FRtext{2.190833}
      \def\JJtext{0.006145}
    \else
      \def\FRtext{2.055895}
      \def\JJtext{-0.006145}
    \fi
  \fi

  \clip (-0.42,-0.42) rectangle (1.42,1.42);
  \draw[step=0.25,gray!18,very thin]
    (-0.5,-0.5) grid (1.5,1.5);
  \draw[gray!38,thin] (-0.42,0) -- (1.42,0);
  \draw[gray!38,thin] (0,-0.42) -- (0,1.42);

  \ifnum#3=0
    \draw[gray!72,thin,->] (0,0) -- (\aa,1)
      node[pos=0.56,right=2pt,text=gray!75] {$R$};
    \draw[gray!60,thin]
      (0.18,0) arc[start angle=0,end angle=90,radius=0.18];
    \node[anchor=south west,text=gray!72]
      at (0.10,0.11) {$\theta$};
  \else
    \draw[gray!72,thin,->] (0,0) -- (\aa,1)
      node[pos=0.58,above left=1pt,text=gray!75] {$R$};
    \draw[gray!60,thin]
      (0.20,0) arc[start angle=0,end angle=45,radius=0.20];
    \node[anchor=south west,text=gray!72]
      at (0.18,0.055) {$\theta$};
  \fi

  \fill[blue!70!black] ( \b, \b) circle (0.018);
  \fill[blue!70!black] ( \b,-\b) circle (0.018);
  \fill[blue!70!black] (-\b, \b) circle (0.018);
  \fill[blue!70!black] (-\b,-\b) circle (0.018);

  \draw[blue!35,dashed,thin]
    (-\b,\b) -- (\b,-\b);
  \draw[blue!70!black,very thick]
    (-\b,-\b) -- (\b,\b);

  \draw[blue!70!black,fill=white,line width=0.7pt]
    (0,0) circle (0.023);

  \fill[orange!90!black] (\aa+\b,1+\b) circle (0.018);
  \fill[orange!90!black] (\aa+\b,1-\b) circle (0.018);
  \fill[orange!90!black] (\aa-\b,1+\b) circle (0.018);
  \fill[orange!90!black] (\aa-\b,1-\b) circle (0.018);

  \ifnum\taustate=1
    \draw[orange!40,dashed,thin]
      (\aa-\b,1+\b) -- (\aa+\b,1-\b);
    \draw[orange!90!black,very thick]
      (\aa-\b,1-\b) -- (\aa+\b,1+\b);
  \else
    \draw[orange!40,dashed,thin]
      (\aa-\b,1-\b) -- (\aa+\b,1+\b);
    \draw[orange!90!black,very thick]
      (\aa-\b,1+\b) -- (\aa+\b,1-\b);
  \fi

  \draw[orange!90!black,fill=white,line width=0.7pt]
    (\aa,1) circle (0.023);

  \draw[gray!58,thin]
    (-0.42,-0.42) rectangle (1.42,1.42);

  \node[
    anchor=south east,
    align=right,
    font=\scriptsize,
    text=black!82
  ] at (1.36,-0.36) {%
    $J_R\sigma\tau=\JJtext$\\
    $J_R=\Jtext$\\
    $F_R(\sigma,\tau)=\FRtext$};

\end{scope}%
}


\def\taustate{1}
\selectorpanel{0}{2.05}{0}
\selectorpanel{2.05}{2.05}{1}

\def\taustate{-1}
\selectorpanel{0}{0}{0}
\selectorpanel{2.05}{0}{1}

\node[font=\small\bfseries,text=gray!72]
  at (0.50,3.61) {$\theta=90^\circ$};

\node[font=\small\bfseries,text=gray!72]
  at (2.55,3.61) {$\theta=45^\circ$};

\node[
  font=\small\bfseries,
  text=orange!90!black,
  anchor=east
] at (-0.50,2.55) {$\tau=+1$};

\node[
  font=\small\bfseries,
  text=orange!90!black,
  anchor=east
] at (-0.50,0.50) {$\tau=-1$};

\node[
  font=\small\bfseries,
  text=blue!70!black,
  anchor=west
] at (-0.42,4.00) {$\sigma=+1$ fixed};

\end{tikzpicture}

\caption{Four exact two-selector configurations for the interaction
$F_R(\sigma,\tau)$ in \eqref{eq:Fst}. The first selector (blue) is
fixed in state $\sigma=+1$, while the second selector (orange) has
state $\tau=+1$ in the top row and $\tau=-1$ in the bottom row.
The left column has an axis-aligned center-difference vector $R$ with
$\theta=90^\circ$, whereas the right column has $R=(1,1)$ and
$\theta=45^\circ$. Thick colored segments denote the selected
diagonals and dashed segments the alternative selector states.
Each panel reports the exact cross-selector Riesz $2$-energy
$F_R(\sigma,\tau)$, the Fourier coupling $J_R$ from
\eqref{eq:fourier-coeff}, and its spin-spin contribution
$J_R\sigma\tau$. The display uses $b=0.12$ only to make the selector
geometry visible; the reported values are evaluated from the exact
interaction \eqref{eq:Fst}, not from the multipole approximation.}
\label{fig:two-selector-gallery}
\end{figure}

Figure~\ref{fig:two-selector-gallery} gives a concrete illustration of
the angular sign change underlying Lemma~\ref{lem:multipole}.  With the
first selector fixed at $\sigma=+1$, an axis-aligned center difference
has $J_R<0$, so the spin-spin contribution $J_R\sigma\tau$ is smaller
for $\tau=+1$ and the interaction favors equal selector states.  When
the second selector is displaced to the $45^\circ$ direction, the sign
reverses.  Then $J_R>0$, and the lower-energy state has $\tau=-1$.  Thus the
same four-point selector geometry yields the ferromagnetic interaction
used along consistency trees and, after changing only the direction of
$R$, the antiferromagnetic interaction used in the objective boxes.
The displayed values of $F_R(\sigma,\tau)$ are computed directly from
\eqref{eq:Fst}; the comparatively large illustrative value $b=0.12$ is
used only to make the geometry visible and is not part of the
asymptotic parameter regime of the reduction.

\begin{lemma}[Uniform angular multipole law]\label{lem:multipole}
If $t\le10^{-8}$, then
\begin{align}
 J_R&=\bigl(384\sin^2(2\theta)-160\bigr)b^4r^{-6}+\mathcal R_J,\label{eq:Jangle}\\
 h_R^{(1)}=h_R^{(2)}
 &=16\sin(2\theta)b^2r^{-4}
   +128\sin(2\theta)b^4r^{-6}+\mathcal R_h,\label{eq:hangle}
\end{align}
with
\begin{equation}\label{eq:rembound}
 \abs{\mathcal R_J}+\abs{\mathcal R_h}
 \le10^6b^5r^{-7}.
\end{equation}
Consequently
\begin{equation}\label{eq:uniformJ}
 \abs{J_R}\le225b^4r^{-6},\qquad
 \abs{h_R^{(i)}}\le17b^2r^{-4}.
\end{equation}
For an axis-aligned center vector,
\begin{equation}\label{eq:axis}
 h_R^{(1)}=h_R^{(2)}=0\quad\text{exactly},\qquad
 J_R<-159b^4r^{-6}.
\end{equation}
For a $45^\circ$ center vector,
\begin{equation}\label{eq:diagJ}
 J_R>223b^4r^{-6}.
\end{equation}
\end{lemma}

\begin{proof}
For one term in \eqref{eq:Fst}, put $v=\eta u_\tau-\epsilon u_\sigma$ and
\[
 z=2t\widehat R\cdot v+t^2\norm v^2,
 \qquad \widehat R=R/r.
\]
Then
\[
 \norm{R+bv}^{-2}=r^{-2}(1+z)^{-1}.
\]
Since $\norm v\le2\sqrt2$, we have $|z|<6t$.  Expanding $(1+z)^{-1}=1-z+z^2-z^3+z^4+\cdots$ through fourth order and applying the Fourier projection gives
\begin{align*}
 [t^4]J_R/r^{-2}
 &=1536\cos^2\theta\sin^2\theta-160
 =384\sin^2(2\theta)-160,\\
 [t^2]h_R^{(i)}/r^{-2}
 &=32\cos\theta\sin\theta=16\sin(2\theta),\\
 [t^4]h_R^{(i)}/r^{-2}
 &=256\cos\theta\sin\theta=128\sin(2\theta).
\end{align*}
The expansion is carried out in Appendix~\ref{app:multipole}; here we bound the remainder explicitly.  Write $a=\widehat R\cdot v$ and $c=\norm v^2$, so that $z=2ta+t^2c$ with $\abs a\le2\sqrt2$ and $c\le8$, and
\[
 (1+z)^{-1}=\sum_{n=0}^{4}(-z)^n+\rho,
 \qquad
 \abs\rho\le\frac{\abs z^5}{1-\abs z}\le\frac{(6t)^5}{1-6t}<7777\,t^5 .
\]
The polynomial part splits into the terms of degree at most four in $t$, which produce the displayed coefficients, and the terms of degree five to eight, namely $-(6ac^2t^5+c^3t^6)$ from $-z^3$ and $32a^3ct^5+24a^2c^2t^6+8ac^3t^7+c^4t^8$ from $z^4$.  For $t\le10^{-8}$ their absolute value is below
\[
 \bigl(6\cdot2\sqrt2\cdot64+32\cdot16\sqrt2\cdot8\bigr)t^5+(512+12288)t^6+8\cdot2\sqrt2\cdot512\,t^7+4096\,t^8<6900\,t^5 .
\]
Hence each of the sixteen point-pair terms deviates from its degree-four polynomial by less than $(7777+6900)t^5<14700\,t^5$ in units of $r^{-2}$.  A Fourier coefficient in \eqref{eq:fourier-coeff} is one quarter of a signed sum of the four values $F_R(\sigma,\tau)$, each of which is a sum of four point-pair terms, so $\abs{\mathcal R_J}<4\cdot14700\,t^5r^{-2}$ and likewise for $\mathcal R_h$.  Together
\[
 \abs{\mathcal R_J}+\abs{\mathcal R_h}<1.2\cdot10^5\,t^5r^{-2}\le10^6b^5r^{-7},
\]
which proves \eqref{eq:rembound}.  The bound is far from sharp: the substitution $(\epsilon,\eta)\to(-\epsilon,-\eta)$ maps $a\to-a$ and fixes $c$, so $F_R$ is an even function of $t$, every odd-order projection vanishes identically, and the true remainder is $O(t^6r^{-2})$.  The cruder uniform bound suffices for the reduction.

At $t\le10^{-8}$ the remainder is at most $0.01b^4r^{-6}$ at the spin-spin scale.  The angular coefficient ranges from $-160$ to $224$, giving \eqref{eq:uniformJ}.  For an axis direction, reflection interchanges the two diagonal states and cancels both one-spin coefficients exactly, while the spin-spin coefficient is negative.  At $45^\circ$ the leading coefficient is $224$, so the coupling is positive.  The field bound follows from \eqref{eq:hangle}.
\end{proof}

\begin{remark}
The sign separation is the essential point.  An axis pair is a field-free ferromagnetic consistency edge, while a $45^\circ$ pair is antiferromagnetic.  At $s=2$ the complete coefficients are rational functions of the rational center coordinates and $b$.
\end{remark}

\section{Exact rational first-order field compensation}\label{sec:field}

The $45^\circ$ interaction and other non-axis pairs produce one-spin fields.  We cancel their leading total locally by perturbing only the $+$ diagonal of each selector:
\[
 c_i\pm b(1,1)
 \quad\longrightarrow\quad
 c_i\pm(b+\eta_i)(1,1),
\]
while the $-$ diagonal remains $c_i\pm b(1,-1)$.

The internal two-state field coefficient becomes
\begin{equation}\label{eq:B-eta}
 B(\eta)=\frac1{16}\left[(b+\eta)^{-2}-b^{-2}\right],
\end{equation}
so
\begin{equation}\label{eq:beta}
 \beta:=B'(0)=-\frac18b^{-3},
 \qquad |\beta|^{-1}=8b^3.
\end{equation}

Let $G_i$ be the complete unperturbed one-spin field at cell $i$, summed over all other cells.  Since center distances are at least one, Lemma~\ref{lem:multipole} gives
\begin{equation}\label{eq:Gbound}
 |G_i|\le17Mb^2.
\end{equation}
We will later define a common objective coefficient $\lambda$ with
\begin{equation}\label{eq:lambda-range-preview}
 19b^4<\lambda<37b^4.
\end{equation}
Set target fields
\begin{equation}\label{eq:targetfields}
 g_i=\begin{cases}
 \lambda,&i\text{ is the root cell of a source-vertex tree},\\
 0,&\text{otherwise}.
 \end{cases}
\end{equation}
The first-order correction is
\begin{equation}\label{eq:eta-star}
 \eta_i^*=\frac{g_i-G_i}{\beta}.
\end{equation}
Because every unperturbed pair energy is rational, $G_i$, $\lambda$, $\beta$, and hence $\eta_i^*$ are rational exactly.

\begin{lemma}[Rational field compensation]\label{lem:field}
Assume
\begin{equation}\label{eq:b-choice-abstract}
 b\le\frac1{10^8M^3}.
\end{equation}
Taking $\eta_i=\eta_i^*$ yields rational perturbed coordinates such that
\begin{enumerate}[label=(\roman*)]
\item $|\eta_i|<150Mb^5$ for every $i$;
\item $|\eta_i|<10^{-6}b$, so the selector-forcing bounds remain valid;
\item the total absolute residual one-spin field obeys
\begin{equation}\label{eq:field-resid}
 \sum_i|\widetilde G_i-g_i|<\frac{\lambda}{100};
\end{equation}
\item the total absolute change of all spin-spin Fourier coefficients obeys
\begin{equation}\label{eq:J-pert}
 \sum_{i<j}|\widetilde J_{ij}-J_{ij}|<\frac{\lambda}{100}.
\end{equation}
\end{enumerate}
All quantities can be computed exactly by rational arithmetic in polynomial time.
\end{lemma}

\begin{proof}
From \eqref{eq:Gbound}, \eqref{eq:lambda-range-preview}, and $|\beta|^{-1}=8b^3$,
\[
 |\eta_i^*|
 \le8b^3(17Mb^2+37b^4)
 <150Mb^5.
\]
Under \eqref{eq:b-choice-abstract}, this is far below $10^{-6}b$.

For $|\eta|\le b/100$, Taylor's theorem applied to \eqref{eq:B-eta} gives the conservative quadratic estimate
\begin{equation}\label{eq:Bquad}
 |B(\eta)-\beta\eta|\le b^{-4}\eta^2.
\end{equation}
All inter-cell point distances stay above $0.99$.  Since $|\nabla_x\norm{x}^{-2}|=2\norm{x}^{-3}$, changing two diagonal radii by $\eta_i,\eta_j$ changes any two-spin Fourier coefficient by at most
\begin{equation}\label{eq:Lipschitz}
 100(|\eta_i|+|\eta_j|),
\end{equation}
where the round constant absorbs the point-pair and Fourier sums.

The total nonlinear internal-bias residual is at most
\[
 M b^{-4}(150Mb^5)^2=22500M^3b^6.
\]
The total variation of all pair couplings is at most
\[
 100\sum_{i<j}(|\eta_i|+|\eta_j|)
 <15000M^3b^5,
\]
and, because every pair carries two one-spin coefficients $h^{(1)}_{ij}$ and $h^{(2)}_{ij}$, the total variation of all pair fields is at most twice that, $30000M^3b^5$.  Using $\lambda>19b^4$ and \eqref{eq:b-choice-abstract}, all three quantities are much smaller than $\lambda/100$: under \eqref{eq:b-choice-abstract} one has $M^3b\le10^{-8}$ and $M^3b^2\le10^{-16}$, so $22500M^3b^6+30000M^3b^5<4\cdot10^{-4}b^4<0.19b^4<\lambda/100$.  This proves \eqref{eq:field-resid} and \eqref{eq:J-pert} with large slack.

Finally, when $s=2$ every squared distance and every energy term is rational.  Therefore the finite Fourier sums defining $G_i$ and $\lambda$ are rational, and \eqref{eq:eta-star} uses only rational operations.  The bit lengths remain polynomial because the construction contains only polynomially many terms whose rational coordinates have polynomial bit length.
\end{proof}

\section{Planar orthogonal routing}\label{sec:routing}

We now construct the selector centers.  The source graph is planar and has maximum degree three.  Polynomial-time algorithms produce orthogonal grid drawings for maximum-degree-four planar graphs with polynomial area and a polynomial number of bends; for example Papakostas and Tollis give linear-time area-efficient constructions \cite{PapakostasTollis1998}.  We only need the coarse consequence stated next.

\begin{lemma}[Coarse orthogonal drawing]\label{lem:orthogonal}
From a planar cubic graph $G$ one can compute in polynomial time a crossing-free orthogonal drawing on the integer grid with total Manhattan edge length $L_0=\operatorname{poly}(n)$ and a polynomial number of bends.
\end{lemma}

Set
\begin{equation}\label{eq:Lbar}
 \overline L=L_0+n+1,
 \qquad
 S=10^5\overline L^2,
 \qquad
 D=S/100=1000\overline L^2.
\end{equation}
Scale the orthogonal drawing by $S$.  Distinct coarse grid lines are now separated by $S=100D$.

Around every source vertex and bend reserve a feature box of radius $3D$.  On every source edge choose one straight scaled segment and reserve an objective box of longitudinal size $6D$ around an interior point.  Every nonzero grid segment has length at least $S=100D$ after scaling.  If two reserved local features would otherwise use overlapping portions of the same corridor, extend the relevant straight segment by a constant multiple of $D$ before placing the unit-spaced selectors.  The extension remains inside the crossing-free corridor, whose separation from unrelated corridors is $100D$.  It therefore creates no crossing, and because only a polynomial number of segments are extended by polynomial length, all coordinates and total routing length remain polynomial.

After the objective box is inserted, the original source edge is severed.  The half-edge routed from a source vertex $v$ to the corresponding objective box becomes one branch of a tree $T_v$.  Hence each $T_v$ is topologically a three-leaf star with one root at $v$ and three internally disjoint rectilinear branches.  Along every branch place selector centers at unit spacing.  At a bend there is a selector center at the bend point, so both incident consistency edges have length one.

The number $M$ of selector cells obeys the coarse bound
\begin{equation}\label{eq:Mbound}
 M\le10\overline L S.
\end{equation}
Indeed the total scaled edge length is $L_0S$, and all local detours and objective prefixes add only $O(nD)\le O(nS)$ length.

After the center set is complete, fix
\begin{equation}\label{eq:b-final}
 \displaystyle b=\frac1{10^8M^3}.
\end{equation}
This ensures every application of Lemma~\ref{lem:multipole}, Lemma~\ref{lem:selector-force}, and Lemma~\ref{lem:field}.

\subsection{A complete vertex gadget}
Figure~\ref{fig:vertex-gadget} shows how the local pieces fit together around one cubic source vertex.  The root selector represents the source spin and carries its compensated target field.  Three unit-spaced axis-aligned branches propagate that spin through bends to three distinct source-edge objective boxes.  Only one objective box is expanded in the figure; the other two are represented by terminal boxes.  The important topological property is that deleting any consistency edge and taking the component away from the root reaches at most one objective box.

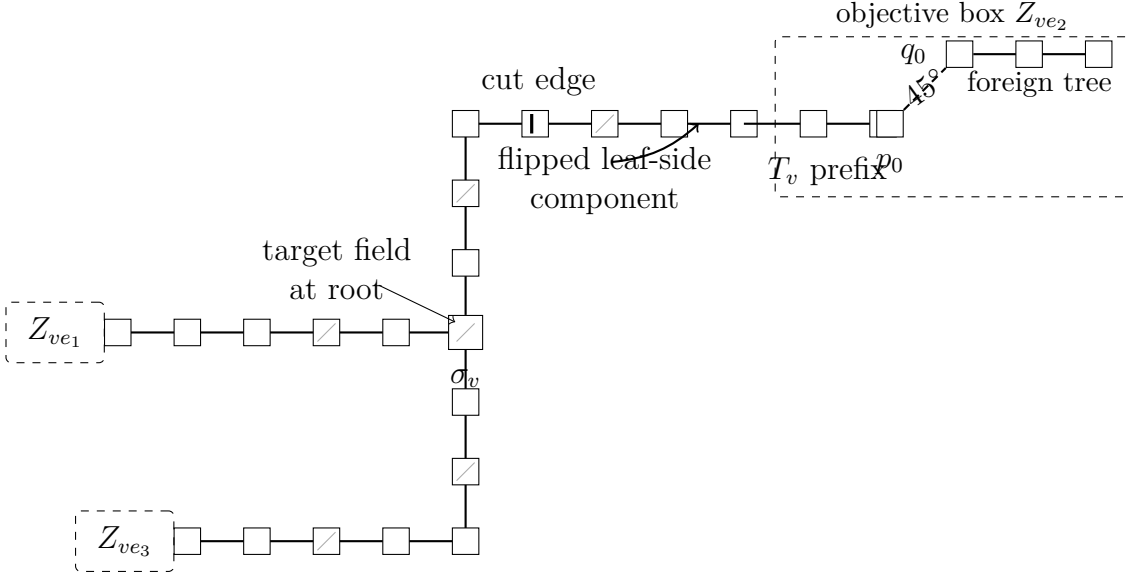
\begin{figure}[htbp]
\centering
\begin{tikzpicture}[x=0.92cm,y=0.92cm]
  \tikzset{
    sel/.style={draw,minimum size=4.5mm,inner sep=0pt,fill=white},
    smallsel/.style={draw,minimum size=3.5mm,inner sep=0pt,fill=white},
    obj/.style={draw,dashed,rounded corners=2pt,inner sep=4pt},
    wire/.style={thick},
    aux/.style={gray!65}
  }

  \node[sel] (r) at (0,0) {};
  \node[below=3pt of r] {$\sigma_v$};

  \draw[wire] (r) -- (-5,0);
  \draw[wire] (r) -- (0,3) -- (4,3);
  \draw[wire] (r) -- (0,-3) -- (-4,-3);

  \foreach \x in {-1,-2,-3,-4,-5} {\node[smallsel] at (\x,0) {};}
  \foreach \y in {1,2,3} {\node[smallsel] at (0,\y) {};}
  \foreach \x in {1,2,3,4} {\node[smallsel] at (\x,3) {};}
  \foreach \y in {-1,-2,-3} {\node[smallsel] at (0,\y) {};}
  \foreach \x in {-1,-2,-3,-4} {\node[smallsel] at (\x,-3) {};}

  \draw[aux] (-0.13,-0.13) -- (0.13,0.13);
  \draw[aux] (-2.13,-0.13) -- (-1.87,0.13);
  \draw[aux] (-0.13,1.87) -- (0.13,2.13);
  \draw[aux] (1.87,2.87) -- (2.13,3.13);
  \draw[aux] (-0.13,-2.13) -- (0.13,-1.87);
  \draw[aux] (-2.13,-3.13) -- (-1.87,-2.87);

  \node[obj,minimum width=1.3cm,minimum height=0.8cm] at (-5.9,0) {$Z_{ve_1}$};
  \node[obj,minimum width=1.3cm,minimum height=0.8cm] at (-4.9,-3) {$Z_{ve_3}$};

  \draw[obj] (4.45,1.95) rectangle (9.55,4.25);
  \node[font=\small] at (7.0,4.55) {objective box $Z_{ve_2}$};
  \draw[wire] (4,3) -- (6.1,3);
  \foreach \x in {5,6} {\node[smallsel] at (\x,3) {};}
  \node[smallsel] (p0) at (6.1,3) {};
  \node[below=2pt of p0] {$p_0$};

  \node[smallsel] (q0) at (7.1,4) {};
  \node[left=3pt of q0] {$q_0$};
  \draw[wire] (q0) -- (9.1,4);
  \foreach \x in {8.1,9.1} {\node[smallsel] at (\x,4) {};}
  \draw[densely dashed,thick] (p0) -- (q0);
  \node[rotate=36] at (6.60,3.52) {$45^\circ$};
  \node[below] at (5.2,2.68) {$T_v$ prefix};
  \node[font=\small] at (8.25,3.58) {foreign tree};

  \draw[very thick] (0.95,2.86) -- (0.95,3.14);
  \node[above] at (1.05,3.25) {cut edge};
  \draw[->,thick] (2.1,2.45) to[bend right=18] (3.35,3.0);
  \node[align=center] at (2.0,2.18) {flipped leaf-side\\component};

  \draw[->] (-1.25,0.68) -- (-0.15,0.15);
  \node[align=center] at (-1.85,0.92) {target field\\at root};
\end{tikzpicture}
\caption{Complete degree-three consistency gadget for one source vertex $v$.  Selector cells lie at unit spacing on the rectilinear tree, and each leaf terminates in one objective box.  The expanded box shows the unique closest $45^\circ$ terminal pair $(p_0,q_0)$.  The marked cut illustrates that a leaf-side component flip meets at most one objective box.  Formal separations are set by $D$ and $S$.}
\label{fig:vertex-gadget}
\end{figure}

\section{The standard objective box}\label{sec:objective}

Inside each objective box use the same template, up to translation, rotation by $90^\circ$, and reflection.  In horizontal orientation the two source trees terminate in selector-center prefixes
\begin{equation}\label{eq:obj-rays}
 p_i=(-i,0),\qquad
 q_j=(1+j,1),
 \qquad 0\le i,j\le D.
\end{equation}
The distinguished pair is
\[
 p_0=(0,0),\qquad q_0=(1,1),
\]
whose center vector has angle $45^\circ$ and length $\sqrt2$.  The prefixes extend away from one another.  Beyond $p_D$ and $q_D$, each branch reconnects to its coarse corridor; all reconnecting features are at distance at least $D$ from the opposite tree.

\paragraph{Remark (geometric fit of the objective box).}
The template above fits inside the reserved part of an orthogonal source-edge corridor.  After scaling, every nonzero straight segment of the grid drawing has length at least $S=100D$, so we may reserve a subsegment of length $6D$ for the objective box.  The $p$ prefix occupies length $D$ on one side of the distinguished pair and the $q$ prefix occupies length $D$ on the other side, with a transverse offset of one unit.  At the far end of the $q$ prefix, a unit vertical step returns to the original corridor.  This reconnecting step is at distance at least $D$ from every center of the opposite prefix and is therefore remote for cross-tree bookkeeping.  Inside its own consistency tree it is \emph{not} an ordinary bend: the $q$ prefix and the continuation along the original corridor are two parallel arms at transverse distance one, and the pairs $(q_{D-1},(1+D,0))$ and $(q_D,(2+D,0))$ have $45^\circ$ center vectors although they lie in the same tree.  Section~\ref{sec:consistency} charges these arm-to-arm pairs by a separate lattice sum \eqref{eq:Usum}.  If the free side of the corridor is opposite, the template is reflected.  Since the original orthogonal drawing is planar and distinct corridors are separated by $100D$, the unit offset remains inside the reserved corridor and creates no crossing.

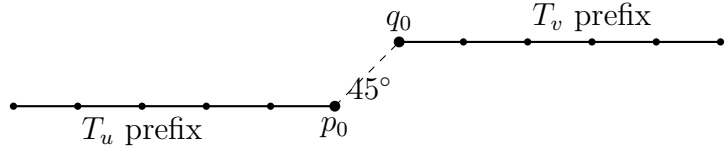
\begin{figure}[ht]
\centering
\begin{tikzpicture}[scale=0.85]
  \draw[thick] (-5,0)--(0,0);
  \draw[thick] (1,1)--(6,1);
  \fill (0,0) circle (2.4pt) node[below] {$p_0$};
  \fill (1,1) circle (2.4pt) node[above] {$q_0$};
  \foreach \x in {-1,-2,-3,-4,-5} \fill (\x,0) circle (1.6pt);
  \foreach \x in {2,3,4,5,6} \fill (\x,1) circle (1.6pt);
  \draw[dashed] (0,0)--(1,1);
  \node at (0.55,0.28) {$45^\circ$};
  \node[below] at (-3,0) {$T_u$ prefix};
  \node[above] at (4,1) {$T_v$ prefix};
\end{tikzpicture}
\caption{Standard objective box.  The unique closest cross-tree center pair is $(p_0,q_0)$; all other local cross-tree pairs form a summable halo.}
\label{fig:objective}
\end{figure}

For the unperturbed selectors define
\begin{equation}\label{eq:lambda-def}
 \lambda=\sum_{i=0}^{D}\sum_{j=0}^{D}J_{q_j-p_i}.
\end{equation}
Rotations and reflections preserve this spin-spin coefficient, so every source edge uses the same~$\lambda$.

\begin{lemma}[Positive objective coefficient]\label{lem:objective}
The objective coefficient satisfies
\begin{equation}\label{eq:lambda-range}
 19b^4<\lambda<37b^4,
\end{equation}
and the complete absolute spin-spin coupling load between the two objective prefixes is less than $37b^4$.
\end{lemma}

\begin{proof}
The distinguished vector $(1,1)$ has $r=\sqrt2$ and angle $45^\circ$.  By \eqref{eq:diagJ},
\[
 J_{q_0-p_0}>223b^4(\sqrt2)^{-6}>27.8b^4.
\]
For every other pair,
\[
 \norm{q_j-p_i}^2=(i+j+1)^2+1.
\]
Putting $k=i+j+1$, there are exactly $k$ nonnegative pairs $(i,j)$ with this value of $k$.  Hence
\begin{align*}
 \sum_{(i,j)\ne(0,0)}|J_{q_j-p_i}|
 &\le225b^4\sum_{k=2}^\infty k\,k^{-6}\\
 &=225b^4(\zeta(5)-1)\\
 &<8.4b^4.
\end{align*}
Thus $\lambda>27.8b^4-8.4b^4>19b^4$.  The distinguished term is at most $225b^4(\sqrt2)^{-6}=28.125b^4$, so the complete absolute load is below $36.6b^4<37b^4$, which also gives the upper bound on~$\lambda$.
\end{proof}

\begin{remark}
The terminal pair is not literally isolated.  It is the unique closest $45^\circ$ pair, and the absolute halo is smaller than its positive coefficient.  This is the invariant needed in the reduction.
\end{remark}

\section{Consistency trees and component-flip dominance}\label{sec:consistency}

Axis-aligned unit neighbors in a consistency tree have, before field perturbation,
\begin{equation}\label{eq:wire-main}
 J_{\rm nn}<-159b^4.
\end{equation}
They therefore prefer equal spins.

For a cut of a straight unit-spaced path, the number of cross-cut pairs at path distance $k$ is $k$.  Excluding the nearest pair,
\begin{equation}\label{eq:straight-tail}
 \Lambda_{\rm straight}
 \le225b^4\sum_{k=2}^\infty k^{-5}
 <8.4b^4.
\end{equation}
At a right-angle sector define
\begin{equation}\label{eq:Qsum}
 Q_6=\sum_{i,j\ge1}(i^2+j^2)^{-3}.
\end{equation}
Since $i^2+j^2\ge2ij$,
\[
 Q_6\le\frac18\zeta(3)^2<0.181.
\]
For opposite collinear rays,
\begin{equation}\label{eq:Psum}
 P_6=\sum_{i,j\ge1}(i+j)^{-6}
 =\zeta(5)-\zeta(6)<0.020.
\end{equation}
\paragraph{Position of the cut on a straight run.}
Let the cut lie on a straight run of unit-spaced cells at path distance $a\ge1$ from one end of the run.  A perpendicular segment attached at that end contributes cross-cut pairs at squared distances $i^2+j^2$ with $i\ge a$ and $j\ge1$, hence at most
\begin{equation}\label{eq:sector-offset}
 \sum_{i\ge a}\sum_{j\ge1}(i^2+j^2)^{-3}
 \le\sum_{i\ge a}\int_0^\infty\frac{dy}{(i^2+y^2)^3}
 =\frac{3\pi}{16}\sum_{i\ge a}i^{-5}
 \le\frac{3\pi}{16}\Bigl(a^{-5}+\frac{a^{-4}}4\Bigr)
 \le a^{-4}\qquad(a\ge2),
\end{equation}
and in every case at most $Q_6$, since the pairs form a subset of the full sector.  A collinear segment attached at that end contributes $\sum_{i\ge a,\,j\ge1}(i+j)^{-6}\le P_6$, which is smaller still.  Because reserved feature boxes of radius $3D$ and objective boxes of length $6D$ do not overlap, and the objective prefixes use only the central $2D+1$ cells of their box, every straight run between two consecutive features of a tree (root, bend, objective-box step, or leaf) has length at least $4D$, with two exceptions treated separately below: the $q$ prefix of length $D$ and the unit step of the objective box.  Hence for a cut on a long run at least one end is at path distance $a\ge2D$ and contributes less than $(2D)^{-4}<10^{-3}$ by \eqref{eq:sector-offset}, while the near end is the root, a bend, or a leaf.

At the root, the run containing the cut meets the two remaining branches, which are either both perpendicular to it or one perpendicular and one opposite-collinear; at a bend it meets one perpendicular run; a leaf exposes nothing inside the tree (its cross-tree load is charged in Section~\ref{sec:objective}).  Charging the straight tail, two perpendicular sectors, one collinear sector, and $10^{-3}$ for the far end therefore covers every cut on a long run with slack:
\begin{equation}\label{eq:localtree}
 \Lambda_{\rm tree}^{\rm local}
 <225b^4(0.037+2\cdot0.181+0.020+0.001)
 <95b^4.
\end{equation}

\paragraph{The objective-box U-turn.}
Inside the objective-box template of Section~\ref{sec:objective} the tree $T_v$ contains the prefix cells $q_j=(1+j,1)$, $0\le j\le D$, the step cell $(1+D,0)$, and the continuation cells $(1+D+m,0)$, $m\ge1$, along the original corridor.  The prefix and the continuation are two parallel arms at transverse distance one joined by the unit step.  This is not a single right-angle bend: the two arms are not adjacent segments, they are not at distance $D$ from each other, and the pairs $(q_{D-1},(1+D,0))$ and $(q_D,(2+D,0))$ have $45^\circ$ center vectors and therefore carry an antiferromagnetic coupling of leading size $28b^4$ each although they lie in the same tree.  We charge all arm-to-arm pairs separately.  A prefix cell $(1+D-i,1)$ and a continuation-side cell $(1+D+j,0)$, $i,j\ge0$, are at squared distance $(i+j)^2+1$.  Summing over all such pairs except the unit step itself, which is the consistency edge under repair whenever it crosses the cut, gives
\begin{equation}\label{eq:Usum}
 U_6=\sum_{\substack{i,j\ge0\\(i,j)\ne(0,0)}}\bigl((i+j)^2+1\bigr)^{-3}
 =\sum_{k\ge1}(k+1)(k^2+1)^{-3}<0.281.
\end{equation}
For a cut on the prefix, on the step, or on the continuation, the cross-cut same-tree pairs consist of pairs on one straight arm, bounded by the straight tail, together with a subset of the arm-to-arm pairs counted in $U_6$; the leaf end $q_0$ exposes nothing further inside the tree, and the far end of the continuation is a bend at path distance at least $4D$, which contributes less than $10^{-3}$ by \eqref{eq:sector-offset}.  Hence a cut in the U-turn region has local same-tree load below $225b^4(0.037+0.281+0.001)<72b^4$, which is inside \eqref{eq:localtree}.  The two $45^\circ$ pairs account for $56b^4$ of this budget.

We use the uniform bound
\begin{equation}\label{eq:localtree-round}
 \Lambda_{\rm tree}^{\rm local}<95b^4
\end{equation}
for every cut of every consistency tree.

\paragraph{Remark (elementary certification of the lattice sums).}
The numerical constants used above require no numerical lattice-sum evaluation.  The straight-tail constant follows from
\[
 \sum_{k=2}^\infty k^{-5}
 \le \sum_{k=2}^{6} k^{-5}+\int_6^\infty x^{-5}\,dx
 <0.03678+0.00020
 <0.037;
\]
note that stopping the finite sum at $k=5$ and integrating from $5$ gives $0.03706$, which would not certify the constant.  For the right-angle and opposite-collinear contributions, the estimates already displayed above give
\[
 Q_6\le \frac18\zeta(3)^2<0.181,
 \qquad
 P_6=\zeta(5)-\zeta(6)<0.020.
\]
For the U-turn sum, $(k+1)(k^2+1)^{-3}\le(k+1)k^{-6}=k^{-5}+k^{-6}$ gives
\[
 U_6\le\sum_{k=1}^{5}(k+1)(k^2+1)^{-3}+\sum_{k\ge6}\bigl(k^{-5}+k^{-6}\bigr)
 <0.2794+0.0004<0.281.
\]
Together these elementary bounds imply \eqref{eq:localtree} and the U-turn bound with slack and hence certify the rounded budget \eqref{eq:localtree-round}.  The actual values are $\zeta(5)-1=0.03693$, $Q_6=0.1474$, $P_6=0.01958$, and $U_6=0.27964$.

\subsection{Remote interactions}\label{sec:remote}

Call a pair of cells \emph{local} if both lie on the same straight run, on two runs sharing a bend or the root, on the two arms of one objective-box U-turn, or in the two prefixes of one objective box; these are exactly the pairs charged in \eqref{eq:straight-tail}--\eqref{eq:Usum} and in Lemma~\ref{lem:objective}.  By the $100D$ coarse-grid spacing and the $3D$ feature boxes, every other pair of cells is at center distance at least $D$.  Thus
\begin{align}
 \Lambda_{\rm rem}
 &\le225M^2b^4D^{-6}\nonumber\\
 &\le225(10\overline LS)^2b^4(S/100)^{-6}\nonumber\\
 &=2.25\cdot10^{16}\overline L^2S^{-4}b^4\nonumber\\
 &<10^{-3}b^4.\label{eq:remote}
\end{align}
The last inequality uses $S=10^5\overline L^2$.

\subsection{Normalization lemma}

Here \emph{normalization} does not mean geometric rescaling.  It is a sequence of cardinality-preserving, strictly energy-decreasing repair moves that places an arbitrary feasible Riesz subset into the intended normal form.  The selector-forcing lemma has already established the first part of that form: every selector contains exactly one of its two valid diagonal pairs and hence represents one spin.  The remaining task in this section is to remove consistency defects inside each tree by component flips.  We call a tree $T_v$ \emph{monochromatic} if all of its selectors realize the same spin, either all $+1$ or all $-1$.  Once $T_v$ is monochromatic, the whole geometric tree can be decoded unambiguously as the single source spin $\sigma_v$.  The lemma below therefore supplies the second normalization step and is the soundness bridge from valid selector states to the discrete Ising state space simulated by the construction.

Apply Lemma~\ref{lem:field} with target fields \eqref{eq:targetfields}.  Thus dummy selectors have target field zero and only the root of each $T_v$ has target field $\lambda$.

\begin{lemma}[Planar consistency-tree normalization]\label{lem:normalize}
After selector forcing, every spin configuration can be transformed by strictly energy-decreasing, cardinality-preserving component flips into one in which every tree $T_v$ is monochromatic.  Each repair has coupling-load margin greater than $25b^4$.
\end{lemma}

\begin{proof}
Root $T_v$ at its source-vertex selector.  If a tree edge $e$ has opposite endpoint spins, delete $e$ and flip all selector spins in the component not containing the root.  This makes the nearest-neighbor product on $e$ equal to $+1$.  The favorable unperturbed coupling magnitude exceeds $159b^4$ by \eqref{eq:wire-main}.

The same-tree local adverse load is below $95b^4$ by \eqref{eq:localtree-round}, which covers cuts on straight runs, at bends, at the root, and in the objective-box U-turn.  Since $T_v$ is topologically a three-leaf star, the leaf-side component reaches at most one objective box; its complete absolute objective load is below $37b^4$.  All remaining unperturbed spin-spin interactions contribute less than $10^{-3}b^4$ by \eqref{eq:remote}.

The intended root field is not flipped.  Lemma~\ref{lem:field} bounds the total residual field by $\lambda/100<0.37b^4$ and the total coupling perturbation by the same amount.  Charging the latter once against the favorable edge and once against all adverse couplings is conservative.  Thus the actual favorable edge has magnitude greater than $158.6b^4$, while the adverse load is below
\[
 95+37+0.001+0.37+0.37<132.75
\]
in units of $b^4$.  The margin exceeds $25b^4$.  An Ising term changes by twice its coupling magnitude under a one-sided spin flip, so the Riesz energy strictly decreases.  Repetition terminates with every $T_v$ monochromatic.
\end{proof}

\section{Normalized Hamiltonian}\label{sec:normalized}

After Lemma~\ref{lem:normalize}, write $\sigma_v$ for the common spin of all selectors in $T_v$.  Classify the exact two-state terms as follows:
\begin{enumerate}[label=(C\arabic*)]
\item state-independent selector and pair constants;
\item same-tree spin-spin terms, for which every product equals $1$;
\item root fields, whose targets contribute $\lambda\sum_v\sigma_v$;
\item objective-box spin-spin terms, each contributing $\lambda\sigma_u\sigma_v$ before perturbation;
\item residual fields, coupling changes caused by field perturbation, and remote inter-tree spin-spin terms.
\end{enumerate}
Absorb classes (C1) and (C2) into a baseline $E_0$ and denote class (C5) by $R(\sigma)$.  Lemma~\ref{lem:field} contributes less than $2\lambda/100$.  By \eqref{eq:remote} and $\lambda>19b^4$, the remote term is less than $10^{-4}\lambda$.  Hence
\begin{equation}\label{eq:Rbound}
 |R(\sigma)|<\frac{\lambda}{20}.
\end{equation}

\begin{proposition}[Barahona normal form]\label{prop:normalform}
Every normalized selector configuration has energy
\begin{equation}\label{eq:normalform}
 E_2(S_\sigma)
 =E_0+\lambda H_B(\sigma)+R(\sigma)
 =E_0+\lambda\left(\sum_{uv\in E}\sigma_u\sigma_v+\sum_v\sigma_v\right)+R(\sigma),
\end{equation}
where $19b^4<\lambda<37b^4$ and $|R(\sigma)|<\lambda/20$.
\end{proposition}

\begin{proof}
The exact Fourier decomposition \eqref{eq:fourier} applies to every selector pair, including the rationally perturbed cells.  Classes (C1) through (C4) give the baseline and source Hamiltonian; all deviations are charged to class (C5).  The displayed remainder bound follows from Lemma~\ref{lem:field} and \eqref{eq:remote}.
\end{proof}

\section{Exact rational threshold and correctness}\label{sec:threshold}

Let $(G,K)$ be the planar cubic independent-set instance and let $H_0$ be given by \eqref{eq:H0}.  Define
\begin{equation}\label{eq:Tstar}
 \tau=E_0+\lambda(H_0+2).
\end{equation}
At exponent two this threshold is rational exactly; Section~\ref{sec:encoding} gives the encoding argument.

\begin{lemma}[Completeness]\label{lem:complete}
If $G$ has an independent set of size at least $K$, then the constructed Riesz instance has a size-$2M$ subset of energy below $\tau$.
\end{lemma}

\begin{proof}
By \eqref{eq:source-yes}, choose $\sigma$ with $H_B(\sigma)\le H_0$ and select the corresponding diagonal throughout each tree $T_v$.  Proposition~\ref{prop:normalform} gives
\[
 E_2(S_\sigma)\le E_0+\lambda H_0+\frac{\lambda}{20}<E_0+\lambda(H_0+2)=\tau.
\]
\end{proof}

\begin{lemma}[Soundness]\label{lem:sound}
If the constructed Riesz instance has a size-$2M$ subset of energy at most $\tau$, then $G$ has an independent set of size at least $K$.
\end{lemma}

\begin{proof}
Apply Lemma~\ref{lem:selector-force} until every cell selects one diagonal, then Lemma~\ref{lem:normalize} until every $T_v$ is monochromatic.  Energy only decreases.  Let $\sigma$ be the resulting source assignment.

If $\alpha(G)\le K-1$, then $H_B(\sigma)\ge H_0+4$ by \eqref{eq:source-no}.  Therefore
\[
 E_2(S_\sigma)
 \ge E_0+\lambda(H_0+4)-\frac{\lambda}{20}
 >E_0+\lambda(H_0+2)=\tau,
\]
a contradiction.  Hence $\alpha(G)\ge K$.
\end{proof}

\section{Polynomial encoding and NP membership}\label{sec:encoding}

The unperturbed selector centers have integer coordinates.  The selector radius
\[
 b=(10^8M^3)^{-1}
\]
is rational with $O(\log M)$ bits.  Lemma~\ref{lem:field} uses the exact rational corrections $\eta_i=(g_i-G_i)/\beta$, so every candidate point belongs to $\Q^2$.

The number of selector cells satisfies
\[
 M\le10\overline L S=10^6\overline L^3,
\]
which is polynomial in the source size.  The constructed point set has $4M$ points and the requested subset size is $k=2M$.

For $s=2$, every pair energy is rational exactly.  If $p,q\in\Q^2$, then
\[
 \norm{p-q}^{-2}
 =\frac1{(p_1-q_1)^2+(p_2-q_2)^2}\in\Q.
\]
Hence all unperturbed Fourier coefficients, $G_i$, $\lambda$, the perturbed coordinates, the baseline $E_0$, and the threshold $\tau$ in \eqref{eq:Tstar} are rational numbers computable by polynomially many exact rational operations.  Their numerator and denominator bit lengths remain polynomial because only polynomially many rational terms of polynomial bit length are combined.

This also proves membership in NP.  A certificate is a subset $S$ of cardinality $k$.  Its $O(k^2)$ inverse-square terms can be summed exactly as rationals and compared with the rational threshold $\tau$ in polynomial time.

\begin{proof}[Proof of Theorem~\ref{thm:main}]
Compute the planar orthogonal drawing, build the selector-center trees and objective boxes, set $b$ by \eqref{eq:b-final}, compute the exact rational field perturbations from Lemma~\ref{lem:field}, and output all four candidate points of each selector together with $k=2M$ and the exact rational threshold \eqref{eq:Tstar}.  Lemmas~\ref{lem:complete} and \ref{lem:sound} give NP-hardness.  The preceding verifier gives membership in NP.  Therefore $\RSSP(2,2)$ is NP-complete.
\end{proof}

\section{Discussion}\label{sec:discussion}

The planar reduction uses a quadrupolar selector rather than a lifted binary pair.  The two selector states have equal mass and zero dipole moment, so the first spin-spin distinction appears at order $b^4$.  For the inverse-square kernel this yields the $r^{-6}$ angular law
\[
 384\sin^2(2\theta)-160,
\]
which is negative on coordinate axes and positive at $45^\circ$.  Axis-aligned chains therefore enforce consistency, while diagonal objective boxes reproduce Barahona's antiferromagnetic source edges.

The external-field source is equally important.  Non-axis selector pairs generate one-spin terms of order $b^2r^{-4}$; rather than trying to eliminate them geometrically, the construction sums them and applies a tiny local diagonal perturbation.  At $s=2$ this compensation is especially clean because all coefficients are rational exactly.

The exponent-two result is arithmetically especially clean.  The reduction has rational coordinates, an exact rational threshold, and a polynomial-time exact verifier.  The theorem is consequently NP-completeness, although the title emphasizes the NP-hardness statement that motivates the reduction.

\paragraph{Outlook.}
The geometric mechanism is not intrinsically restricted to two ambient dimensions.  Since $\R^2$ embeds isometrically into $\R^d$ for every fixed $d>2$, the NP-hardness result at $s=2$ carries immediately to all higher fixed dimensions by appending zero coordinates.  The exponent direction is also structurally favorable.  For fixed $s>2$, the selector penalty becomes stronger relative to inter-cell interactions and the quadrupolar spin--spin tail decays as $r^{-(s+4)}$, faster than the $r^{-6}$ decay used here; the axis/diagonal sign separation of the leading multipole coefficient persists.  In particular, for fixed odd integer exponents $s>2$ the same reduction architecture preserves NP-hardness after the quantitative constants and field compensation are retuned, with the separation and tail estimates becoming easier rather than harder.  What is lost is the exact rational arithmetic special to $s=2$: for odd $s$, rational coordinates generally produce algebraic irrational pair energies.  Thus the exact rational verifier of Section~\ref{sec:encoding} no longer establishes membership in NP directly, and NP-completeness for those exponents requires a separate analysis of exact algebraic comparison.  Together with the polynomial-time result on the line for fixed $s>0$ \cite{EmmerichLine2026}, this suggests a broader complexity landscape governed jointly by dimension and exponent.

A second direction is algorithmic.  NP-hardness makes fast Riesz-energy subset-selection heuristics and approximation methods particularly important for large point sets.  Falc\'on-Cardona, Ju\'arez, M\'arquez-Vega, and Emmerich \cite{FalconCardonaJuarezMarquezVegaEmmerich2026} provide an important first scalable step: their Riesz-energy subset-selection methods use greedy inclusion and iterative replacement and are designed for near-linear practical scaling on large Pareto archives.  A natural next goal is to connect such approximation-oriented heuristics with the present complexity theory by developing provable approximation guarantees, instance-dependent lower bounds or certificates, and geometric conditions under which substantially faster algorithms are possible.  The planar and low-dimensional regimes are especially attractive for such work because they combine application relevance with additional geometric structure that may be exploitable algorithmically.

\section*{Acknowledgements}
The author thanks Ksenia Pereverdieva and André Deutz for their support during earlier attempts at this proof, including discussions of alternative source problems and reduction strategies. These explorations helped inform the development of the present approach.

\section*{Tool and computational resource disclosure}
AI-assisted tools, including ChatGPT (OpenAI) and Claude (Anthropic), were used during development of this manuscript to support adversarial proof audits, symbolic cross-checks, and editorial review.  In particular, an independent soundness audit of version~2 was carried out with Claude Fable~5 (Anthropic); it recomputed the multipole coefficients, the lattice-sum constants, the objective coefficient $\lambda$, and the component-flip loads in exact rational arithmetic, and it identified the presentational gaps corrected in the present version (the explicit U-turn sum \eqref{eq:Usum}, the explicit remainder bound in Lemma~\ref{lem:multipole}, and the certification arithmetic in Section~\ref{sec:consistency}).  The audit report and the exact-arithmetic scripts are included in the companion repository cited in Appendix~\ref{app:cube-example}.  Automated tools were used as research aids and not as sources of mathematical authority.  The author reviewed and verified the mathematical arguments, references, and presentation and retains full responsibility for the correctness and adequacy of the work.

\appendix

\section{Worked eight-vertex reduction example}\label{app:cube-example}

This appendix gives a concrete finite instance of the reduction.  Its purpose is pedagogical.  It shows how a small planar cubic source graph is expanded into selector trees, objective boxes, rational candidate points, and a Riesz threshold, and it checks the relevant energy hierarchy directly.  The example is not used in the proof of the general theorem, whose constants are deliberately much more conservative.

A companion repository is available at\\
\url{https://github.com/emmerichmtm/RieszEnergyFromIsing3DPlanarCubicReduction}.\\
It contains (i) the Python script that builds the finite instance of this appendix, recomputes all fields and couplings from the rational coordinates, and audits every consistency-tree edge; (ii) a symbolic (SymPy) derivation of the multipole coefficients of Lemma~\ref{lem:multipole}; (iii) exact-rational (\texttt{fractions}) checks of the multipole law, the remainder, the lattice-sum constants, $\lambda$, and the worst-case component-flip loads for every cut type including the objective-box U-turn; (iv) two Lean~4/Mathlib files, which are arithmetic sanity checks of displayed inequalities (\texttt{norm\_num}, \texttt{linarith}) and an exhaustive \texttt{native\_decide} enumeration of $Q_3$, not a formalization of the reduction; and (v) the independent AI soundness audit mentioned in the tool disclosure.

\paragraph{Why eight vertices rather than seven.}
A cubic graph cannot have seven vertices.  By the handshake lemma, the sum of the degrees is twice the number of edges and is therefore even, whereas a $3$-regular graph on seven vertices would have degree sum $21$.  We therefore use the cube graph $Q_3$, the smallest familiar planar cubic example with enough structure to display all parts of the construction.  Label its vertices $0,\ldots,7$ and take
\begin{align*}
E={}&\{01,12,23,30,45,56,67,74,04,15,26,37\}.
\end{align*}
The outer cycle is $0,1,2,3$, the inner cycle is $4,5,6,7$, and the remaining four edges are spokes.

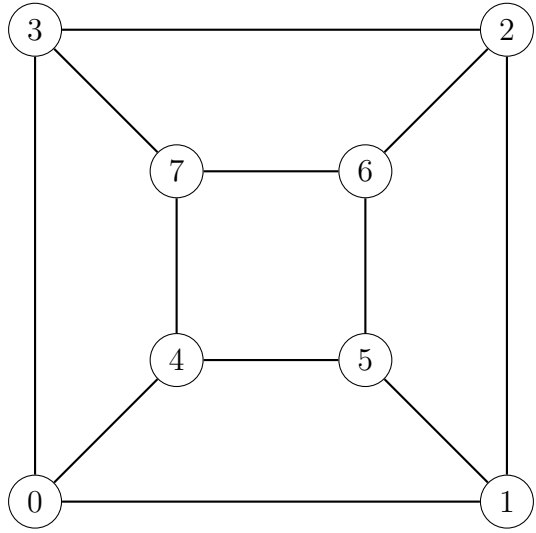
\begin{figure}[htbp]
\centering
\begin{tikzpicture}[scale=0.78,
  vtx/.style={circle,draw,minimum size=7mm,inner sep=0pt,fill=white},
  ed/.style={thick}]
  \node[vtx] (v0) at (-4,-4) {0};
  \node[vtx] (v1) at ( 4,-4) {1};
  \node[vtx] (v2) at ( 4, 4) {2};
  \node[vtx] (v3) at (-4, 4) {3};
  \node[vtx] (v4) at (-1.6,-1.6) {4};
  \node[vtx] (v5) at ( 1.6,-1.6) {5};
  \node[vtx] (v6) at ( 1.6, 1.6) {6};
  \node[vtx] (v7) at (-1.6, 1.6) {7};
  \draw[ed] (v0)--(v1)--(v2)--(v3)--(v0);
  \draw[ed] (v4)--(v5)--(v6)--(v7)--(v4);
  \draw[ed] (v0)--(v4);
  \draw[ed] (v1)--(v5);
  \draw[ed] (v2)--(v6);
  \draw[ed] (v3)--(v7);
\end{tikzpicture}
\caption{The planar cubic source graph used in the worked example is the cube graph $Q_3$.  Every source vertex has degree three, so after an orthogonal drawing each source vertex is replaced by one root selector and three rectilinear branches.}
\label{fig:cube-source}
\end{figure}

\paragraph{Independent set and Ising source.}
Exhaustive enumeration of the $2^8=256$ source states gives $\alpha(Q_3)=4$.  For example, $\{0,2,5,7\}$ and $\{1,3,4,6\}$ are independent sets of size four.  Under the convention $\sigma_v=+1$ for a selected source vertex, the Barahona Hamiltonian has minimum value $-12$, attained by those two maximum independent sets, while the next energy level is $-8$.  Thus the four-unit source gap is already visible in this tiny instance.

\paragraph{Expanding one source vertex into a selector tree.}
An explicit orthogonal drawing is used for this finite example.  Each source edge is cut once in its interior.  The half-edge incident to $v$ becomes one branch of the consistency tree $T_v$.  Because $Q_3$ is cubic, each $T_v$ is only a three-leaf star with one root and three nonbranching chains, with right-angle bends where the drawing requires them.  There is no recursive branching.  The branch length records only how far the corresponding edge has to travel in the planar drawing.

The human-scale drawing used here has 392 selector cells in total.  The four outer source vertices use 61 cells each and the four inner vertices use 37 cells each.  Hence the resulting Riesz instance has
\[
 4M=1568\quad\text{candidate points},\qquad k=2M=784.
\]
Every selector center is integral.  Its four candidate points are obtained by adding the two diagonal pairs described in Section~\ref{sec:selector}; after field compensation all printed coordinates are finite decimals and therefore rational.

\paragraph{One objective box per source edge.}
For every one of the 12 source edges, choose a straight part of its orthogonal route and replace a short middle segment by the standard diagonal terminal geometry.  For readability this finite example uses objective-prefix length $D=2$.  The two terminal centers differ by a vector of the form $(\pm1,\pm1)$, so the distinguished interaction is at $45^\circ$, while the cells behind the terminals retreat along axis-aligned prefixes.  Thus the finite example has exactly the same local logic as the general construction.  Axis neighbors propagate a source spin, and a diagonal terminal pair represents a source edge.

\paragraph{A finite energy hierarchy.}
For a readable finite-scale check, take $b=1/500$ and evaluate the resulting inverse-square interactions with high precision.  This $b$ is much larger than the theorem's conservative asymptotic choice, but the finite instance can be checked directly.  The resulting values are

\begin{center}
\begin{tabular}{@{}ll@{}}
\toprule
quantity & finite example value \\
\midrule
selector cells $M$ & $392$ \\
candidate points / chosen points & $1568$ / $784$ \\
selector-forcing left side $(1/16)b^{-2}$ & $15625$ \\
selector-forcing right side $4M$ & $1568$ \\
objective coefficient $\lambda/b^4$ & $29.270918\ldots$ \\
maximum $|\eta_i|/b$ & $6.07\cdot10^{-10}$ \\
effective root fields divided by $\lambda$ & $0.99983\ldots$ to $1.00115\ldots$ \\
source-edge couplings divided by $\lambda$ & $0.98606\ldots$ to $0.98689\ldots$ \\
largest nonedge coupling divided by $\lambda$ & $1.17\cdot10^{-4}$ \\
minimum component-flip margin & $66.87\,b^4$ \\
maximum normalized remainder $|R|/\lambda$ & $0.168$ \\
\bottomrule
\end{tabular}
\end{center}

The last row exceeds the bound $|R|<\lambda/20$ of Proposition~\ref{prop:normalform}.  This is not a counterexample: the proposition assumes the asymptotic choice $b=(10^8M^3)^{-1}$ and the separation scale $D=1000\overline L^2$, whereas this instance uses $b=1/500$ and $D=2$ so that it can be drawn and inspected.  The relevant check for the finite instance is the explicit threshold separation reported below, which holds with the larger remainder.

Two points of this table are worth emphasizing.  First, the selector-forcing inequality already holds by almost an order of magnitude, so a size-$784$ subset cannot improve its energy by using an invalid occupancy or a side pair.  Second, the component-flip check is not a local heuristic.  For every consistency-tree edge, delete that edge, take the component away from the root, and sum the absolute value of \emph{all} fields and couplings that cross the resulting cut.  Even the worst cut retains more than $66b^4$ of favorable margin.  Thus every nonmonochromatic selector tree admits a strictly improving repair in this finite instance.

\paragraph{Collapsing the 392 selectors back to eight spins.}
After the preceding check, each tree can be represented by a single source spin.  Recomputing the complete effective eight-spin Hamiltonian from all Riesz pair interactions also includes the unwanted far field.  Intended graph-edge coefficients are close to the common objective scale $\lambda$, while nonedge coefficients are roughly four orders of magnitude smaller.  Over all 256 source-spin assignments the deviation from the ideal scaled Barahona energy is at most $0.168\lambda$.

For the independent-set target $K=4$, the ideal source levels are $-12$ and at least $-8$, so the midpoint scaled threshold is $-10$.  In the actual finite Riesz geometry, the two ground states have scaled energy at most
\[
 -11.8336\ldots,
\]
whereas every state from the next source level has scaled energy at least
\[
 -7.9220\ldots.
\]
The midpoint therefore remains safely separating.  All coordinates and the corresponding finite-decimal threshold are rational.

\section{Detailed derivation of the multipole coefficients}
\label{app:multipole}

This appendix gives a more detailed derivation of the two-selector
interaction formula and of the angular coefficients appearing in
Lemma~\ref{lem:multipole}.  The purpose is to make explicit how the
ordinary inverse-square Riesz interaction between individual selected
points gives rise to the effective spin interaction
\[
A_R+h_R^{(1)}\sigma+h_R^{(2)}\tau+J_R\sigma\tau,
\]
and why the orientation of the center vector $R$ enters through
$\sin(2\theta)$ and $\sin^2(2\theta)$.

\paragraph{Geometry of two selector cells.}

Consider two unperturbed selector cells with centers
\[
c_1,c_2\in\mathbb R^2,
\]
and define their center-difference vector by
\begin{equation}
\label{eq:app-Rdef}
R=c_2-c_1.
\end{equation}
Thus $R$ points from the center of the first selector to the center
of the second selector.  It is important that $R$ connects the
\emph{selector centers}; the centers themselves are not points of the
Riesz instance.

Recall that
\[
u_+=(1,1),
\qquad
u_-=(1,-1).
\]
If the first selector has spin $\sigma\in\{-1,+1\}$, its two selected
Riesz points are
\begin{equation}
\label{eq:app-xeps}
x_\epsilon=c_1+\epsilon b u_\sigma,
\qquad
\epsilon\in\{-1,+1\}.
\end{equation}
Similarly, if the second selector has spin $\tau$, its two selected
Riesz points are
\begin{equation}
\label{eq:app-yeta}
y_\eta=c_2+\eta b u_\tau,
\qquad
\eta\in\{-1,+1\}.
\end{equation}

Hence the selector center is surrounded by four candidate Riesz
points, but after selector forcing only the two points on the diagonal
specified by the spin are selected.

For a particular pair of selected points, subtracting
\eqref{eq:app-xeps} from \eqref{eq:app-yeta} gives
\begin{align}
y_\eta-x_\epsilon
&=
(c_2+\eta b u_\tau)
-
(c_1+\epsilon b u_\sigma)
\nonumber\\
&=
(c_2-c_1)
+
b(\eta u_\tau-\epsilon u_\sigma)
\nonumber\\
&=
R+b(\eta u_\tau-\epsilon u_\sigma).
\label{eq:app-displacement}
\end{align}

At exponent $s=2$, one pair of Riesz points contributes the
inverse-square potential
\[
\norm{y_\eta-x_\epsilon}^{-2}.
\]
Each selector contains two selected points, so there are four
cross-selector point pairs.  Summing their interactions gives
\[
F_R(\sigma,\tau)
=
\sum_{\epsilon,\eta\in\sgnspin}
\norm{R+b(\eta u_\tau-\epsilon u_\sigma)}^{-2},
\]
which is precisely \eqref{eq:Fst}.  Thus \eqref{eq:Fst} is simply the
ordinary Riesz interaction between the two selected diagonals, written
in a form that keeps the selector spins explicit.

\paragraph{Why the direction of $R$ produces trigonometric terms.}

Write
\begin{equation}
\label{eq:app-polarR}
r=\norm R,
\qquad
\widehat R=\frac{R}{r}
=(\cos\theta,\sin\theta),
\qquad
t=\frac br.
\end{equation}
Here $\theta$ is the angle of the center vector $R$ relative to the
positive horizontal coordinate axis.

For example,
\[
R=(1,0)
\quad\Longrightarrow\quad
\theta=0,
\]
whereas
\[
R=(1,1)
\quad\Longrightarrow\quad
\theta=\frac{\pi}{4}.
\]

The inverse-square Riesz kernel itself is radial and therefore depends
only on distance.  The angular dependence arises because the selector
is not rotationally symmetric: its two states occupy the two diagonal
directions
\[
u_+=(1,1),
\qquad
u_-=(1,-1).
\]
Taking scalar products with $\widehat R$ gives
\begin{align}
\widehat R\cdot u_+
&=\cos\theta+\sin\theta,
\label{eq:app-dotplus}\\
\widehat R\cdot u_-
&=\cos\theta-\sin\theta.
\label{eq:app-dotminus}
\end{align}
Thus the orientation of the line joining the selector centers enters
the energy through ordinary Euclidean dot products.

\paragraph{Normalization of one Riesz term.}

For fixed $\sigma,\tau,\epsilon,\eta$, put
\begin{equation}
\label{eq:app-vdef}
v=\eta u_\tau-\epsilon u_\sigma.
\end{equation}
Then one term in \eqref{eq:Fst} is
\[
\norm{R+bv}^{-2}.
\]
Expanding the squared norm gives
\begin{align}
\norm{R+bv}^{2}
&=
\norm R^2+2bR\cdot v+b^2\norm v^2
\nonumber\\
&=
r^2+2br\,\widehat R\cdot v+b^2\norm v^2
\nonumber\\
&=
r^2
\left(
1+2t\,\widehat R\cdot v+t^2\norm v^2
\right).
\label{eq:app-norm-expand}
\end{align}
Define
\begin{equation}
\label{eq:app-ac}
a=\widehat R\cdot v,
\qquad
c=\norm v^2,
\end{equation}
and
\begin{equation}
\label{eq:app-z}
z=2ta+t^2c.
\end{equation}
Then
\begin{equation}
\label{eq:app-riesz-normalized}
\norm{R+bv}^{-2}
=
r^{-2}(1+z)^{-1}.
\end{equation}

This normalization separates the overall distance scale $r$ from the
small dimensionless quantity
\[
t=\frac br.
\]
The construction chooses $b$ so small that all applications of the
multipole expansion have $t\le10^{-8}$.

Since
\[
\norm{u_\pm}=\sqrt2,
\]
we have
\[
\norm v
\le
\norm{u_\tau}+\norm{u_\sigma}
=
2\sqrt2.
\]
Consequently
\[
|a|\le2\sqrt2,
\qquad
c\le8,
\]
and hence
\[
|z|
\le
4\sqrt2\,t+8t^2
<6t
\]
in the range used in the construction.

\paragraph{Expansion through fourth order.}

Using
\[
(1+z)^{-1}
=
1-z+z^2-z^3+z^4+O(z^5)
\]
and substituting $z=2ta+t^2c$, the terms through order $t^4$ are
\begin{align}
-z
&=-2at-c t^2,
\\
z^2
&=
4a^2t^2+4ac t^3+c^2t^4,
\\
-z^3
&=
-8a^3t^3-12a^2c t^4+O(t^5),
\\
z^4
&=
16a^4t^4+O(t^5).
\end{align}
Therefore
\begin{equation}
\label{eq:app-master-expansion}
\begin{aligned}
(1+z)^{-1}
={}&
1
-2a t
+(4a^2-c)t^2
\\
&+(4ac-8a^3)t^3
\\
&+(c^2-12a^2c+16a^4)t^4
+O(t^5).
\end{aligned}
\end{equation}

Equation~\eqref{eq:app-master-expansion} is the basic local expansion
from which both the one-spin fields and the spin-spin coefficient are
obtained.

\paragraph{Fourier projection onto the binary spins.}

For fixed $R$, the quantity $F_R(\sigma,\tau)$ is a function of two
binary variables.  The four functions
\[
1,\qquad \sigma,\qquad \tau,\qquad \sigma\tau
\]
form a basis for all real-valued functions on
$\{-1,+1\}^2$.  Hence
\[
F_R(\sigma,\tau)
=
A_R+h_R^{(1)}\sigma+h_R^{(2)}\tau+J_R\sigma\tau.
\]
The coefficient $J_R$ is the projection onto the basis function
$\sigma\tau$:
\begin{equation}
\label{eq:app-PJ}
J_R
=
\frac14
\sum_{\sigma,\tau\in\sgnspin}
\sigma\tau\,F_R(\sigma,\tau).
\end{equation}
Similarly,
\begin{align}
h_R^{(1)}
&=
\frac14
\sum_{\sigma,\tau\in\sgnspin}
\sigma\,F_R(\sigma,\tau),
\label{eq:app-Ph1}\\
h_R^{(2)}
&=
\frac14
\sum_{\sigma,\tau\in\sgnspin}
\tau\,F_R(\sigma,\tau).
\label{eq:app-Ph2}
\end{align}

For each fixed pair $(\sigma,\tau)$, $F_R$ contains four choices of
$(\epsilon,\eta)$.  The Fourier projection then sums over the four
spin states $(\sigma,\tau)$.  Thus the coefficient calculation
contains sixteen elementary point-pair terms in total.

For the remaining algebra, abbreviate
\[
x=\cos\theta,
\qquad
y=\sin\theta.
\]
Then
\[
x^2+y^2=1.
\]

Substitute
\[
v=\eta u_\tau-\epsilon u_\sigma
\]
into $a=\widehat R\cdot v$ and $c=\norm v^2$, insert these quantities
into \eqref{eq:app-master-expansion}, and perform the finite sums over
\[
\sigma,\tau,\epsilon,\eta\in\{-1,+1\}.
\]
The odd-order terms cancel under the Fourier projections.  For the
spin-spin coefficient, the order-$t^2$ term also cancels:
\[
[t^2]\,J_R/r^{-2}=0.
\]
Thus the first spin-spin distinction occurs at order $t^4$.

It is useful to display explicitly where the fourth-order coefficient
comes from.  The three terms in
\[
c^2-12a^2c+16a^4
\]
contribute, after the $J_R$ projection,
\begin{align}
\mathcal P_J[c^2]
&=32,
\label{eq:app-Jpiece1}\\
\mathcal P_J[-12a^2c]
&=-192(x^2+y^2),
\label{eq:app-Jpiece2}\\
\mathcal P_J[16a^4]
&=1536x^2y^2.
\label{eq:app-Jpiece3}
\end{align}
Since $x^2+y^2=1$, their sum is
\begin{align}
[t^4]\,J_R/r^{-2}
&=
1536x^2y^2-192+32
\nonumber\\
&=
1536x^2y^2-160.
\label{eq:app-Jxy}
\end{align}

This makes the origin of the integer $-160$ transparent:
\[
-160=32-192.
\]
The coefficient is exact; it is not a numerical approximation.  It
arises from the finite sum over the binary selector and point signs.

For the one-spin coefficients, the order-$t^2$ projection gives
\begin{equation}
\label{eq:app-h2}
[t^2]\,h_R^{(i)}/r^{-2}
=
32xy,
\qquad i=1,2.
\end{equation}
At fourth order,
\begin{align}
[t^4]\,h_R^{(i)}/r^{-2}
&=
1024x^3y+1024xy^3-768xy
\nonumber\\
&=
1024xy(x^2+y^2)-768xy
\nonumber\\
&=
256xy.
\label{eq:app-h4}
\end{align}
The equality of the two one-spin coefficients follows from the
symmetry between the two selectors.

Combining the preceding calculations gives
\begin{align}
J_R/r^{-2}
&=
\left(1536x^2y^2-160\right)t^4
+O(t^5),
\label{eq:app-Jt}\\
h_R^{(i)}/r^{-2}
&=
32xy\,t^2
+
256xy\,t^4
+O(t^5).
\label{eq:app-ht}
\end{align}

\paragraph{Conversion to the angular formulas.}

The standard double-angle identities
\[
2\sin\theta\cos\theta=\sin(2\theta)
\]
and
\[
4\sin^2\theta\cos^2\theta=\sin^2(2\theta)
\]
give
\[
2xy=\sin(2\theta),
\qquad
4x^2y^2=\sin^2(2\theta).
\]
Therefore
\begin{align}
1536x^2y^2-160
&=
384\sin^2(2\theta)-160,
\label{eq:app-Jtrig}\\
32xy
&=
16\sin(2\theta),
\label{eq:app-htrig2}\\
256xy
&=
128\sin(2\theta).
\label{eq:app-htrig4}
\end{align}

Finally,
\[
r^{-2}t^2
=
r^{-2}\left(\frac br\right)^2
=
b^2r^{-4},
\]
whereas
\[
r^{-2}t^4
=
r^{-2}\left(\frac br\right)^4
=
b^4r^{-6}.
\]
Consequently,
\begin{align}
J_R
&=
\bigl(384\sin^2(2\theta)-160\bigr)
b^4r^{-6}
+
O(b^5r^{-7}),
\label{eq:app-J-final}\\
h_R^{(1)}=h_R^{(2)}
&=
16\sin(2\theta)b^2r^{-4}
+
128\sin(2\theta)b^4r^{-6}
+
O(b^5r^{-7}).
\label{eq:app-h-final}
\end{align}
These are the leading terms in
\eqref{eq:Jangle} and \eqref{eq:hangle}.

The proof of Lemma~\ref{lem:multipole} uses the uniform estimate
$|z|<6t$, the geometric-series tail $\abs\rho<7777t^5$, and the explicit
degree-five-to-eight pieces of $z^3$ and $z^4$ to replace the
$O(b^5r^{-7})$ notation by the explicit conservative bound
\[
\abs{\mathcal R_J}+\abs{\mathcal R_h}
<1.2\cdot10^{5}\,t^5r^{-2}
\le
10^6 b^5r^{-7}.
\]
The large numerical constant is used only to obtain a simple uniform
bound; the coefficients $384$, $160$, $16$, and $128$ above are exact.
Since $F_R$ is even in $t$, the remainder is in fact $O(t^6r^{-2})$; an
exact-rational evaluation of \eqref{eq:Fst} at $b=10^{-3}$ and $10^{-4}$
for the center vectors used in the construction gives a coefficient below
$4\cdot10^{3}$ in units of $t^6r^{-2}$, attained at $45^\circ$.

\paragraph{Geometric interpretation of the angular dependence.}

The spin-spin coefficient has leading angular factor
\[
384\sin^2(2\theta)-160.
\]
Its significance becomes particularly clear in the two orientations
used by the reduction.

For an axis-aligned pair of selector centers,
\[
\theta=0
\quad\text{or}\quad
\theta=\frac{\pi}{2},
\]
so
\[
\sin(2\theta)=0.
\]
Hence
\[
J_R
=
-160b^4r^{-6}
+
\mathcal R_J.
\]
The coupling is therefore negative.  Since the corresponding Ising
term is
\[
J_R\sigma\tau,
\]
a negative coefficient favors
\[
\sigma\tau=+1,
\]
that is, equal selector spins.  This is the ferromagnetic interaction
used to propagate a common spin along the axis-aligned consistency
trees.

At an angle of $45^\circ$,
\[
\theta=\frac{\pi}{4},
\]
so
\[
\sin^2(2\theta)
=
\sin^2\left(\frac{\pi}{2}\right)
=
1.
\]
The exact leading angular coefficient is therefore
\[
384-160=224.
\]
Thus
\[
J_R
=
224b^4r^{-6}
+
\mathcal R_J,
\]
which is positive.  A positive coefficient favors
\[
\sigma\tau=-1,
\]
and hence opposite selector spins.  This is the antiferromagnetic
interaction used inside an objective box to encode a source edge.

Notice also that
\[
0\le\sin^2(2\theta)\le1,
\]
so $45^\circ$ maximizes the positive part of the leading angular
coefficient.  It therefore gives the strongest positive interaction
available from this angular law.

For the distinguished terminal pair in the standard objective box,
\[
R=(1,1),
\qquad
r=\sqrt2,
\qquad
\theta=\frac{\pi}{4}.
\]
Since
\[
r^{-6}=(\sqrt2)^{-6}=\frac18,
\]
the leading contribution is
\[
224b^4r^{-6}
=
28b^4.
\]
The uniform remainder estimate then leaves a positive interaction
slightly below this value, which is the quantitative input used later
in Lemma~\ref{lem:objective}.

The same calculation also explains the different orders of the two
effective terms.  The one-spin field begins at order
\[
b^2r^{-4},
\]
whereas the spin-spin distinction begins only at order
\[
b^4r^{-6}.
\]
The latter cancellation reflects the balanced geometry of the
selector.  Both states have the same number of selected points and
zero dipole moment.  The first state-dependent interaction between
two such balanced selectors is therefore quadrupolar, which is why
the spin-spin term appears two orders later.

In summary, the derivation starts from the individual inverse-square interactions and rewrites each relative displacement as $R+b(\eta u_\tau-\epsilon u_\sigma)$.  Writing the direction of $R$ as $\widehat R=(\cos\theta,\sin\theta)$ isolates the angular dependence.  Expanding in the small parameter $t=b/r$ then gives a finite polynomial through fourth order, and Fourier projection onto $1$, $\sigma$, $\tau$, and $\sigma\tau$ extracts the spin-dependent coefficients.  The spin-spin projection yields $J_R=\bigl(384\sin^2(2\theta)-160\bigr)b^4r^{-6}+\mathcal R_J$.  The angular sign change in this coefficient is the feature used by the reduction.
Axis-aligned pairs propagate a spin consistently, while the
$45^\circ$ terminal geometry in an objective box implements the
positive antiferromagnetic source-edge interaction.

\section{Table of frequently used symbols}\label{app:symbols}

The table lists the symbols that recur across several sections, with a brief explanation and the place where each is introduced.  Symbols that are used only locally inside one derivation (for instance the auxiliary quantities $a$, $c$, $z$, $v$, $x=\cos\theta$, $y=\sin\theta$ of the multipole expansion) are explained where they occur.

\begin{center}
\small
\begin{tabular}{@{}>{$}l<{$}p{7.6cm}l@{}}
\toprule
\text{Symbol} & \text{Meaning} & Introduced \\
\midrule
\multicolumn{3}{@{}l}{\emph{Source problem}}\\
G=(V,E),\ n,\ m & planar cubic source graph, $n=|V|$, $m=|E|=3n/2$ & Sec.~\ref{sec:source} \\
K,\ \alpha(G) & independent-set target; independence number & Sec.~\ref{sec:source} \\
\sigma_v\in\{-1,+1\} & Ising spin of source vertex $v$ ($+1$ = selected) & Sec.~\ref{sec:source} \\
H_B(\sigma) & Barahona Hamiltonian $\sum_{uv\in E}\sigma_u\sigma_v+\sum_v\sigma_v$ & Eq.~\eqref{eq:HB} \\
t,\ q & number of selected vertices; number of induced edges & Sec.~\ref{sec:source} \\
H_0 & source threshold $n/2-4K$; the yes/no gap is $4$ & Eq.~\eqref{eq:H0} \\
\midrule
\multicolumn{3}{@{}l}{\emph{Target problem and selector cells}}\\
E_s(S) & Riesz $s$-energy of a point set $S$; here $s=2$ & Sec.~\ref{sec:intro}, Thm.~\ref{thm:main} \\
P,\ k,\ \tau & constructed point set, subset cardinality, rational threshold & Thm.~\ref{thm:main}, Eq.~\eqref{eq:Tstar} \\
c,\ b & selector center; selector scale (half-side of the candidate square) & Sec.~\ref{sec:selector} \\
S_+(c),\ S_-(c) & the two diagonal states of a selector, spins $+1$ and $-1$ & Sec.~\ref{sec:selector} \\
I_j & minimum internal energy of $j$ selected corners of one cell & Sec.~\ref{sec:selector} \\
\Delta_{\rm sel},\ c_{\rm sel} & side-versus-diagonal penalty $c_{\rm sel}b^{-2}$, $c_{\rm sel}=1/8$ & Eq.~\eqref{eq:csel} \\
M & total number of selector cells; $k=2M$ & Eq.~\eqref{eq:k2M} \\
\eta_i,\ \eta_i^* & perturbation of the $+$ diagonal radius of cell $i$; its first-order value & Sec.~\ref{sec:field}, Eq.~\eqref{eq:eta-star} \\
\midrule
\multicolumn{3}{@{}l}{\emph{Two-selector interaction}}\\
u_+,\ u_- & diagonal directions $(1,1)$ and $(1,-1)$ & Sec.~\ref{sec:multipole} \\
R,\ r,\ \theta & center-difference vector, its length, its angle to the $x$-axis & Sec.~\ref{sec:multipole} \\
t=b/r & small expansion parameter (distinct from the vertex count $t$ above) & Sec.~\ref{sec:multipole} \\
F_R(\sigma,\tau) & exact four-pair interaction of two selectors & Eq.~\eqref{eq:Fst} \\
A_R,\ h_R^{(1)},\ h_R^{(2)},\ J_R & Fourier coefficients: constant, one-spin fields, spin-spin coupling & Eq.~\eqref{eq:fourier-coeff} \\
\mathcal R_J,\ \mathcal R_h & multipole remainders, $|\mathcal R_J|+|\mathcal R_h|\le10^6b^5r^{-7}$ & Lemma~\ref{lem:multipole} \\
J_{\rm nn} & coupling of axis-aligned unit neighbors, $<-159b^4$ & Eq.~\eqref{eq:wire-main} \\
\widetilde J_{ij},\ \widetilde G_i & couplings and fields of the perturbed instance & Lemma~\ref{lem:field} \\
\bottomrule
\end{tabular}
\end{center}

\begin{center}
\small
\begin{tabular}{@{}>{$}l<{$}p{7.6cm}l@{}}
\toprule
\text{Symbol} & \text{Meaning} & Introduced \\
\midrule
\multicolumn{3}{@{}l}{\emph{Field compensation}}\\
B(\eta),\ \beta & internal one-spin coefficient of a perturbed cell; $\beta=B'(0)=-b^{-3}/8$ & Eqs.~\eqref{eq:B-eta}--\eqref{eq:beta} \\
G_i & complete unperturbed one-spin field at cell $i$ & Sec.~\ref{sec:field} \\
g_i & target field: $\lambda$ at tree roots, $0$ elsewhere & Eq.~\eqref{eq:targetfields} \\
\lambda & objective coefficient, $19b^4<\lambda<37b^4$ & Eqs.~\eqref{eq:lambda-def}, \eqref{eq:lambda-range} \\
\midrule
\multicolumn{3}{@{}l}{\emph{Routing and gadgets}}\\
L_0,\ \overline L & Manhattan length of the coarse drawing; $\overline L=L_0+n+1$ & Lemma~\ref{lem:orthogonal}, Eq.~\eqref{eq:Lbar} \\
S,\ D & grid scaling factor $10^5\overline L^2$; separation scale $S/100$ & Eq.~\eqref{eq:Lbar} \\
T_v & consistency tree (three-leaf star of selector cells) of source vertex $v$ & Sec.~\ref{sec:routing} \\
p_i,\ q_j & selector-center prefixes of the two trees inside an objective box & Eq.~\eqref{eq:obj-rays} \\
\midrule
\multicolumn{3}{@{}l}{\emph{Energy bookkeeping}}\\
\Lambda_{\rm straight},\ \Lambda^{\rm local}_{\rm tree} & cross-cut coupling loads: straight tail; local same-tree & Eqs.~\eqref{eq:straight-tail}, \eqref{eq:localtree-round} \\
\Lambda_{\rm rem} & total load of all remote (non-local) pairs, $<10^{-3}b^4$ & Eq.~\eqref{eq:remote} \\
Q_6,\ P_6,\ U_6 & lattice sums for a right angle, opposite rays, and the objective-box U-turn & Eqs.~\eqref{eq:Qsum}, \eqref{eq:Psum}, \eqref{eq:Usum} \\
\text{(C1)--(C5)} & classes of terms of a normalized configuration & Sec.~\ref{sec:normalized} \\
E_0,\ R(\sigma) & state-independent baseline; remainder with $|R(\sigma)|<\lambda/20$ & Prop.~\ref{prop:normalform} \\
\bottomrule
\end{tabular}
\end{center}

\end{document}